\documentclass[11pt,aps,pra,onecolumn,tightenlines,amsfonts,floatfix,superscriptaddress]{revtex4}

\usepackage{amssymb}
\usepackage{graphicx}
\usepackage{dcolumn}
\usepackage{bm}
\usepackage{mathbbol}
\usepackage{tikz}

\usepackage{amssymb}          
\usepackage{graphicx}            
\usepackage[latin1]{inputenc} 
\usepackage{amsmath}          
\usepackage{color}
\usepackage{epstopdf}
\usepackage{blkarray}
\usepackage[caption=false]{subfig}

\newtheorem{prop}{Proposition}
\newtheorem{defin}{Definition}
\newtheorem{thm}{Theorem}
\newtheorem{cor}{Corollary}
\newtheorem{lemma}{Lemma}

\newtheorem{remark}{Remark}

\newcommand{\ket}[1]{|#1\rangle}
\newcommand{\bra}[1]{\langle #1|}
\newcommand{\braket}[2]{\langle #1|#2\rangle}

\usepackage{enumerate}
\usepackage{amsmath}
\usepackage{cancel}
\usepackage{mathrsfs}
\usepackage{mathdots}
\usepackage{euscript}
\usepackage{amscd}
\usepackage{placeins}
\usepackage{algorithm2e}
\usepackage{tikz}
\usetikzlibrary{decorations.pathmorphing,arrows,shapes}

\graphicspath{{images/}}

\newtheorem{ex}{Example}

\newcommand{\bmat}{\left[ \begin{matrix}}
\newcommand{\emat}{\end{matrix} \right]}

\newcommand{\Tr}{\mathrm{tr}}

\newcommand{\qed}{\hfill $\Box$ \vskip 2ex}

\newcommand{\tr} {\mbox{\rm tr}}

\newcommand{\Zbb}{\mathbb Z}

\newcommand{\cC}{\mathcal{C}}
\newcommand{\cD}{\mathcal{D}}

\newcommand{\cP}{\mathcal{P}}

\newcommand{\cF}{\mathcal{F}}
\newcommand{\cH}{\mathcal{H}}
\newcommand{\cN}{\mathcal{N}}
\newcommand{\cM}{\mathcal{M}}
\newcommand{\cR}{\mathcal{R}}

\newcommand{\cI}{\mathcal{I}}

\newcommand{\cS}{\mathcal S}

\newcommand{\identity}{{I}}

\newcommand{\Hi}{\mathcal{H}}
\newcommand{\Ei}{\mathcal{E}}
\newcommand{\Fi}{\mathcal{F}}
\newcommand{\Li}{\mathcal{L}}
\newcommand{\Ni}{\mathcal{N}}

\newcommand{\beq}{\begin{equation}}
\newcommand{\eeq}{\end{equation}}
\newcommand{\beqa}{\begin{eqnarray}}
\newcommand{\eeqa}{\end{eqnarray}}
\newcommand{\beqan}{\begin{eqnarray*}}
\newcommand{\eeqan}{\end{eqnarray*}}
\newcommand{\bea}{\begin{eqnarray}}
\newcommand{\eea}{\end{eqnarray}}

\begin{document}
\title{Dissipative Encoding of Quantum Information}

\author{Giacomo Baggio}
\affiliation{\mbox{Dipartimento di Ingegneria dell'Informazione, Universit\`a di Padova, via Gradenigo 6/B, 35131 Padova, Italy}}

\author{Francesco Ticozzi} 
\affiliation{\mbox{Dipartimento di Ingegneria dell'Informazione, Universit\`a di Padova, via Gradenigo 6/B, 35131 Padova, Italy}} 
\affiliation{\mbox{Department of Physics and Astronomy, Dartmouth College, 6127 Wilder Laboratory, Hanover, NH 03755, USA}}

\author{Peter D. Johnson}
\affiliation{Zapata Computing, Inc., 501 Massachusetts Avenue, Cambridge, MA 02139, USA}

\author{Lorenza Viola}
\affiliation{\mbox{Department of Physics and Astronomy, Dartmouth College, 6127 Wilder Laboratory, Hanover, NH 03755, USA}}

\begin{abstract}
We formalize the problem of dissipative quantum encoding, and explore the advantages of using Markovian evolution to prepare a quantum code in the desired logical space, with emphasis on discrete-time dynamics and the possibility of exact finite-time convergence. In particular, we investigate robustness of the encoding dynamics and their ability to tolerate initialization errors, thanks to the existence of non-trivial basins of attraction.  As a key application, we show that for stabilizer quantum codes on qubits, a finite-time dissipative encoder may always be constructed, by using at most a number of quantum maps determined by the number of stabilizer generators. We find that even in situations where the target code lacks gauge degrees of freedom in its subsystem form, dissipative encoders afford nontrivial robustness against initialization errors, thus overcoming a limitation of purely unitary encoding procedures. Our general results are illustrated in a number of relevant examples, including Kitaev's toric code. 
\end{abstract}

\date{\today}
\maketitle

\nopagebreak

\section{Introduction}

Implementing quantum information processing in physical devices requires that abstractly defined quantum information, carried by ideal information units (qubits), be represented using the available degrees of freedom \cite{Nielsen-Chuang,QECBook}. Central in this context is the concept of a {\em quantum code}: while in principle a code, ${\cal C}$, is simply any subset of the physical system's state space, ${\cal H}_P$, the details of the code are essential in determining the precise sense in which the quantum information of interest may be preserved against uncontrolled noise that physical systems are inevitably exposed to \cite{Blume-Kohout2010}. In particular, logically represented (``encoded'') quantum information may be intrinsically immune to the action of a set of errors -- for instance, by virtue of special symmetry properties, such as in decoherence-free subspaces or noiseless subsystems \cite{dfs,Knill2000} -- or, more generally, it may be actively protected through suitable recovery operations -- such as in quantum error-correcting codes \cite{QECBook}. Ultimately, in connection with the accuracy threshold theorem, quantum error correction (QEC) will be key in enabling large-scale fault-tolerant quantum computation, provided the noise is sufficiently well behaved \cite{QECBook,Preskill}.  

Among various approaches for finding quantum error-correcting codes that have been pursued, the {\em stabilizer formalism} has proved to be especially powerful in describing a large and important family of QEC codes and their error-correcting structure in a very compact form \cite{QECBook}. In its standard version, a stabilizer code protects quantum information through a suitable subspace encoding, that is, encoded quantum states are restricted to a {\em subspace} of ${\cal H}_P$ \cite{Gottesman1998}. Within the more general formulation of QEC theory afforded by the subsystem notion \cite{Knill1997,Knill2000} (also later referred to as ``operator QEC,'' OQEC \cite{Oqec}), a stabilizer code encodes the information to be protected in a {\em subsystem of a subspace} of ${\cal H}_P$ \cite{Poulin2005}. Notably, the presence of auxiliary ``gauge degrees'' of freedom in subsystem codes can both lead to simpler error-recovery procedures, with implications for quantum fault-tolerance \cite{Aliferis}, as well as to intrinsic tolerance of the code against additional errors \cite{Hideo2013}.

Clearly, ensuring that the logical information of interest is effectively encoded in the target code is of crucial importance for the proper functioning of QEC itself. Loosely speaking, in practical settings such an encoding task entails transferring the information of interest from a quantum state of some accessible yet unprotected ``upload'' physical subsystem, where it initially resides, to a state in a subspace or subsystem that represents its encoded counterpart. Since the dynamics implementing this transfer must work for each state of the upload qubits, encoding procedures must be devised without making explicit reference to a particular input state. From a control standpoint, designing a quantum encoder amounts to finding dynamics that implement a continuous family of specified one-to-one state transitions, from each state of the upload qubits to its corresponding codeword state -- which is a challenging problem in general. 

Methods for constructing encoding unitary dynamics within the circuit model of quantum computation have been extensively explored for stabilizer codes \cite{Gottesman1998,Cleve}, including their subsystem extensions \cite{Klappenecker2009}. Our interest here is to revisit the encoding problem from the perspective of using {\em engineered dissipative dynamics} \cite{altafini-introduction}, which have gained increasing significance for quantum tasks ranging from robust quantum state preparation, steady-state entanglement generation and cooling \cite{kraus,Ticozzi2012,BTV2012,cooling}, to open quantum system simulation \cite{opensys} and quantum-limited amplification \cite{clerk}. Notably, schemes for achieving dissipative quantum memories \cite{Pastawski2011}, dissipative quantum computation \cite{Verstraete-DQC}, and autonomous QEC \cite{Reiter,Home,Sarlette} have also been put forward, whereas continuous-time Markovian dynamics have been proposed in \cite{Dengis2014} to encode information in a specific albeit  paradigmatic stabilizer code -- Kitaev's toric code on the square lattice \cite{Kitaev2003}. 

In this work, we characterize the general features of {\em dissipative quantum encoders}, and propose a systematic way to construct Markovian dynamics for encoding stabilizer codes, with special emphasis on {\em discrete-time dynamics} and the possibility of exact finite-time convergence \cite{johnson-FTS}. {Our analysis both puts on a rigorous foundation and substantially expands the preliminary account of dissipative encoding in continuous time we provided in \cite{encoding-CDC}.} Aside from its intrinsic appeal as an alternative route to traditional unitary schemes, the use of dissipative encoding dynamics has the potential advantage of supporting non-trivial basins of attraction: formally, the control problem is akin to devising a continuous family of {\em many-to-one} state transitions, one for each state in the target code. This feature can be potentially exploited to tolerate errors and offer more flexibility on the initialization of the physical upload qubits. 

More specifically, the work is organized as follows. After providing some essential background on quantum codes in Sec. \ref{sec:back}, we formalize the general encoding task as a two-step procedure in Sec. \ref{sec:dqe}: the information to be encoded is first {\em initialized} from the logical level in a physical, upload subsystem, and then {\em encoded} into the target code ${\cal C}$ by a suitably engineered quantum evolution on the physical degrees of freedom. When the evolution is dissipative, investigating the general structure of the basin of attraction for the procedure leads naturally, in Sec. \ref{sub:basin}, to examine robustness against faulty initializations, represented by noise maps that act prior to the physical encoding step, and to an existential characterization of noise-tolerant dissipative encoders in terms of {\em compatible subsystem decompositions} (Theorem \ref{thm:noise}). After this general discussion, in Sec. \ref{sub:Markov} we formally introduce the important class of {\em Markovian dissipative encoders} based on both  continuous-time dynamics -- implemented by a semigroup (Lindblad) master equation, as in \cite{Dengis2014} -- or discrete-time dynamics -- implemented by a sequence of quantum maps, comprising a dissipative quantum circuit, in the spirit of \cite{johnson-FTS}.   

In the second part of the paper, we focus on stabilizer codes and establish a number of constructive results. In Sec. \ref{sec:FTE}, we show that discrete-time encoders able to dissipatively prepare the target code subspace in {\em finite time} exists, {using a number of steps determined by the number of stabilizer generators. Notably, we also show that such encoders always have non-empty basins of attractions even under additional constraints that may be relevant to the analysis}, such as invariance of the code and specific forms for the encoding maps and logical operators. While the latter requirement may make the proposed encoders look similar to a stabilizer QEC protocol, a fundamental difference stems from the fact that, in our setting, the ``errors'' that can be tolerated are not specified at the outset, but rather emerge from the form of the code and the initialization subsystems. Likewise, although also for our task the encoding maps entail measurement of stabilizer operators followed by unitary ``correction,'' the latter are chosen using different criteria than they are in QEC. Our construction is exemplified in a number of relevant stabilizer codes in Sec. \ref{sub:examples}. In particular, while some degree of robustness against initialization errors is known to be achievable with unitary dynamics in subsystem codes as long as non-trivial gauge factors can be identified \cite{Klappenecker2009}, our analysis makes it clear that using dissipative encoders may be the \emph{only} way to attain similar robustness for subspace codes, or whenever gauge qubits are not easily identifiable. The construction of a finite-time encoder for Kitaev's toric code is addressed separately in Sec. \ref{sub:tc}, by directly leveraging special geometric features this code enjoys. While both the continuous-time encoding dynamics of \cite{Dengis2014} and the dissipative quantum circuit we propose respect the same locality structure of the underlying stabilizer operators, our construction ensures that the target code space can be reached in a finite number of steps (proportional to the number of physical qubits) with zero error, in principle -- something which is never possible when convergence is exponential. We briefly conclude in Sec. \ref{sec:end}. 

\section{Background}
\label{sec:back}

\subsection{Quantum codes, subsystems, and the isometric approach}

Mathematically, in order to specify a code that carries quantum information associated to an abstract, {\em logical} quantum system with Hilbert space $\Hi_L$, it is necessary to identify a subset of states of the {\em physical} system, with corresponding Hilbert space $\Hi_P$. In this work, we take both $\Hi_{L}$ and $\Hi_{P}$ to be finite-dimensional, and consider the representation of the full set of density operators on $\Hi_L$, denoted by ${\cal D}(\Hi_L),$ using a subset ${\cal C}\subset{\cal D}(\Hi_P)$, with ${\cal C}$ being the {\em code}. The space of all linear (bounded) operators on, say, $\Hi_P$ is denoted by $\mathcal{B}(\Hi_P)$. Throughout the paper, we shall also write $\cal {A}\simeq\cal{B}$ to denote that two spaces, $\cal {A}, \cal {B}$, or decomposition thereof, are isomorphic. ${A}\simeq{B}$ denotes that operator $A$ is mapped into $B$ via the same isomorphism (that is, in terms of their matrix representatives, they are the same up to a suitable choice of basis).

The first kind of quantum codes that have been discovered and studied, both as active \cite{ShorCode,SteaneCode} or passive \cite{dfs} codes, are {\em subspace codes} \cite{QECBook}. These are associated to sets of states that have support on a subspace ${\cal H}_{\cal C}\equiv {\cal H}_S$ of the physical Hilbert space $\cH_P$, with $\cH_{S}\simeq \cH_L$, so that we can write
$$\cH_P = \cH_{S}\oplus \cH_{R} \simeq \cH_{L}\oplus \cH_{R}. $$
The quantum codewords are then the density operators supported on $\cH_S$, that is, ${\cal C}= \cD(\cH_{S})$, and the summand $\cH_R=\cH_S^\perp$ in this case. The code subspace is chosen so that the action of the intended noise, which is modeled as a completely-positive, trace-preserving (CPTP) map ${\cal M}$, is either trivial (e.g., states in a decoherence-free subspace are invariant under ${\cal M}$) or can be recovered by means of available measurements and correction operations -- namely, there exists a recovery CPTP map ${\cal R}$ such that $({\cal R}\circ {\cal M}) |\psi\rangle\langle\psi |= |\psi\rangle\langle\psi |$, for all $|\psi\rangle \in {\cH}_S$ \cite{Knill1997}. In the simplest setting where the error model ${\cal M}$ corresponds to independent errors on qubits, the {\em distance} $d$ of the code yields the minimum number of single-qubit operations needed to transform a codeword into another (the notion may be generalized to arbitrary error models \cite{Knill2000}). Thus, one may formally view passive codes (${\cal R}_{|{\cal C}} = {\cal I}_S$, where ${\cal I}$ is the identity map) as infinite-distance QEC codes. 

However, subspaces provably do not furnish the most general quantum codes possible. The {\em subsystem principle} for QEC, anticipated in \cite{Knill1997} and established in \cite{Knill2000,Viola2001,Knill2006,Ticozzi2010}, states that any (passive or active) quantum code can be associated to a general {\em subsystem} decomposition:
\begin{equation}
\cH_P = \cH_S \otimes \cH_F \oplus \cH_R \simeq\cH_{L}\otimes \cH_{F}\oplus \cH_{R}.
\label{sub}
\end{equation}
As above, $\Hi_S$ is isomorphic to the space to be encoded but, in a subsystem code, this space appears as a factor, in tensor product with another Hilbert space $\Hi_F,$ representing a {\em gauge} subsystem on which ${\cal M}$ can act without affecting the information encoded in ${\cal C}= \cD(\cH_{S})$. More precisely, we say that a state of the physical system, $\overline\rho\in\cD(\cH_P)$, is {\em initialized in $\cH_S$ with state $\rho\in\cD(\cH_S)$, and gauge state $\tau_{F}\in\cD(\cH_{F})$}, if $\overline\rho\simeq\rho\otimes \tau_{F}\oplus 0_{R},$ where $0_R$ denotes the zero operator on $\cH_R$. In particular, we say that $\overline\rho\in\cD(\cH_P)$ is  initialized in a subsystem pure state if the above equation holds with $\rho=|\psi\rangle\langle\psi |$, for some $|\psi\rangle \in {\cH}_S$. 

When a gauge state $\tau_F=\sum_{j=1}^f p_j\ket{\phi_j}\bra{\phi_j}$ is specified, one can think of a subsystem code as a collection of $f$ orthogonal subspace codes of the form $\Hi_L\otimes{\rm span}\{\ket{\phi_j}\}$, in each of which one initializes a fraction $p_j$ of the total probability. Thus, $\tau_F$ does not carry any logical information, it only specifies how the information is distributed over the set of orthogonal subspace codes. It is then possible to prove that information encoded into subsystem codes is intrinsically robust with respect to changes of the co-factor state, with the restriction of ${\cal M}$ to $\cH_S \otimes \cH_F$ obeying ${\cal M}_{|\cH_S \otimes \cH_F} = {\cal I}_S \otimes {\cal F}$, for some CPTP map ${\cal F}$ on $\cH_F$ \cite{QECBook,Knill2006}. 

In principle, it is possible to construct subsystem codes that extend a given subspace code $\cH_P\simeq\cH_{L}\oplus \cH_{R}.$ The new subsystem is obtaining by identifying $f-1$ isomorphic and mutually orthogonal copies of $\cH_{S}$ inside $\cH_{R}.$  However, in practical cases it may be difficult to find such copies so that they are collectively recoverable after the action of ${\cal M}$. Nonetheless, we can always see a subspace code as a subsystem code with a {\em one-dimensional} gauge co-factor. In light of the above, in the more theoretical part of the paper (Sec. \ref{sec:dqe}), we shall work directly with subsystem codes, and see subspace codes as a particular case.

Conversely, given a subsystem code, it is possible to obtain a subspace code by fixing the gauge, that is, by losing the freedom in the gauge state. By imposing that $\tau_F$ be a specified pure state $\ket{\phi}\bra{\phi}$, the only subspace that is allowed to carry information in $\Hi_L\otimes\Hi_F$ is $\Hi_L'\simeq\Hi_L\otimes{\rm span}\{\ket{\phi}\}.$ We illustrate some of these ideas for the simplest quantum code, the repetition code, which will also be revisited and used as a guiding example in the rest of the paper.

\begin{ex}[The repetition code as a subsystem code]
\label{ex:rep1}
{\em This 3-qubit code is usually described as a subspace code that encodes one qubit protected with respect to independent single bit-flip errors, $\Hi_{\cal C}=\textrm{span}\{\ket{000},\ket{111}\}$ and errors act via ${\cal M} (\rho)= (1-p) \rho+(p/3) \sum_{\ell=1,2,3} X_\ell \rho X_\ell$, $\rho\in\cD(\cH_P)$, where $X_\ell$ denotes the Pauli $X$ acting on qubit $\ell$, $X_1=XII, X_2=IXI, X_3=IIX$, and $p <1/2$ is the error probability per qubit. Revisiting the code in the subsystem picture of Eq. \eqref{sub} allows to clearly identify the action of noise that can be effectively corrected, as well as the required error correction \cite{Knill1997,Viola2001}. The code subspace $\Hi_{\cal C}$ can be associated to a natural subsystem decomposition $\Hi_{P} \simeq \Hi_{L} \otimes \Hi_F = {\mathbb C}^2 \otimes {\mathbb C}^4$, induced by the unitary change of basis $U_L$ defined by
\begin{equation}
U_L\ket{abc}\equiv \ket{x}\otimes \ket{yz}, \qquad a,b,c\in\{0,1\},
\label{rep}
\eeq 
\noindent 
where $x$ is the majority count of the string $abc,$ while $yz$ indicates the binary location in which $abc$ differs from $xxx$, with $00$ indicating no differences.  With respect to this subsystem decomposition, the original code subspace corresponds to $\Hi_{\cal C}\simeq \Hi_{L} \otimes {\rm span}\{|00\rangle\}$, and its codeword states are uniquely associated to the states of the subsystem code ${\cal C}= \{ \rho\otimes \ket{00}\bra{00}\},$ with $\rho\in{\mathcal{D}(\Hi_L)}$. It is easy to see that the action of any noise that does not affect the majority count with respect to the subsystem decomposition \eqref{rep}, namely, of evolutions of the form ${\cal M}=\cI_L \otimes \Fi_F$, can be corrected by a recovery map ${\cal R}$ that resets (``cools'') the co-factor gauge qubits back to $\ket{00}.$ That is, in this representation, the code ${\cal C}$ is fixed -- is a noiseless subsystem \cite{Knill2000} -- under ${\cal M}\circ{\cal R}$. }
\end{ex}

In \cite{Ticozzi2010}, a natural operational interpretation for the subsystem principle is provided, which contributes to clarify connections among various notions of error protection and correction. A quantum code ${\cal C}$ can be identified as the image of a  CPTP map $\Phi$ from the  logical degrees of freedom $\Hi_{L}$ into the physical Hilbert space $\Hi_{P},$  that preserves the \emph{distinguishability} of states. Explicitly, the image of a CPTP map $\Phi$ defines a code if for all $\rho_1,\rho_2\in\mathcal{D}(\Hi_L)$ and $p\in[0,1]$, $$\|p\Phi(\rho_1)-(1-p)\Phi(\rho_2)\|_1=\|p\rho_1-(1-p)\rho_2\|_1,$$
where $\| A \|_1\equiv \Tr [|A|] = \Tr[ \sqrt{A^\dagger A } ]$. This requirement is equivalent to saying that $\Phi$ is a {\em trace-norm isometric embedding} of $\mathcal{B}(\Hi_{L})$ into  $\mathcal{B}(\Hi_P)$.

As shown in \cite{Ticozzi2010}, this isometry property, together with linearity and the CPTP requirements, are both sufficient and necessary to ensure that the image of $\Phi$ is a subsystem encoding: Any $1$-isometric CPTP embedding $\Phi$ induces a subsystem decomposition $\Hi_P\simeq\Hi_L\otimes\Hi_F\oplus\Hi_R$, where $\Phi(\rho)\simeq\rho\otimes\tau_F\oplus 0_R$ for some given $\tau_F\in\mathcal{B}(\Hi_F)$. A quantum code, then, may be described as $\mathcal{C}\equiv\Phi(\mathcal{D}(\Hi_L))$, and its codewords are the states of the form $\rho\otimes\tau_F\oplus 0_R$, for a given $\tau_F$. Preservation of distinguishability is also shown to be a necessary and sufficient requirement for a code undergoing a noisy physical evolution to be perfectly correctable.

\subsection{Basics of stabilizer formalism}
\subsubsection{Stabilizer codes}

Let $\cP_{n}$ denote the $n$-qubit Pauli group \cite{Nielsen-Chuang}.  A {\em stabilizer (subspace) code} {$\cC$ on $(\mathbb{C}^2)^{\otimes n}$ is supported on the common $+1$-eigenspace $\cH_\cS$} of a set of commuting operators $\{S_{k}\}_{k=1}^{r}\subseteq \cP_{n}$, which are called {\em stabilizer operators}.  These operators generate an Abelian subgroup $\cS$ of $\cP_{n}$, the so-called stabilizer (sub)group. Thus, $\cH_\cC =\cH_\cS \equiv \text{span}\{ |\psi\rangle\,|\, S_k|\psi\rangle=|\psi\rangle, \, k=1,\ldots, r\}$, implying that the code space is invariant under the action of $\cS$. Thanks to the properties of Pauli matrices, the dimension of $\cH_\cS$ is $2^{n-r}$: This subspace can then be used to encode $n-r$ logical qubits (also often called ``virtual'' qubits, as they need not be in a direct relation to the physical ones). To this aim, we define a set of {\em logical Pauli operators}, say, $\{\overline{X}_{k},\overline{Z}_{k}\}_{k=1}^{n-r}$, acting on $\cH_\cS$. These operators must commute with all stabilizer generators and therefore belong to the centralizer of $\cS$; thanks to the properties of the stabilizer group, the latter corresponds to the \emph{normalizer} $\cN(\cS).$ For a Pauli subgroup, the latter is the set of operators that leave all the elements of $\cS$ invariant under conjugate (adjoint) action, that is, $\cN(\cS)\equiv \{P \in \cP_n\,|\, PSP^\dag=S, \, \forall S\in \cS\}$. Clearly, $\cS \subset \cN(\cS)$. The operators which are in $\cN(\cS)-\cS$ generate a subgroup of the Pauli group, called the {\em logical subgroup}, since it corresponds to operators that affect non-trivially the information encoded in the logical qubits. 

Given a Pauli subgroup, following \cite{QECBook}, we say that a set of generators is a {\em canonical basis} if it is composed by pairs of operators $\hat X_\ell,\hat Z_\ell$ (virtual $X$ and $Z$ operators) such that they anti-commute, $\{\hat X_\ell,\hat Z_\ell\}=0$, while they commute for $j\neq\ell,$  $[\hat X_\ell,\hat Z_j]=0$, in addition to $[\hat X_\ell,\hat X_j]=0$ and $[\hat Z_\ell,\hat Z_j]=0$, for any pair of indexes. When a subgroup admits a canonical set of generators, it can be seen as a (virtual) qubit system. Being $\cS$ a stabilizer subgroup, it can be shown that its centralizer (hence, its normalizer) is generated by $iI$, $\cS$ itself, and a canonical basis of $n-r$ $\hat X_\ell,\hat Z_\ell$ pairs. The $\ell$-th logical qubit is then naturally associated to the corresponding pair of logical operators $\hat X_\ell,\hat Z_\ell$. It is worth stressing that the choice of basis is, in general, highly not unique. Also, the use of a suitable symplectic representation is especially convenient in allowing one to check for commutativity via simple linear-algebraic manipulations (see also Appendix \ref{sympl}).

As for classical linear codes, the key properties of a (binary) stabilizer code are described as a string of three parameters \cite{QECBook}, $[[n,k,d]],$ where $n$ is the total number of physical qubits, $k=n-r$ the number of logical qubits encoded (thus, dim$(\cH_\cC)=2^k$), and $d$ denotes the distance of the code, which for a stabilizer code is given by the minimum weight of any Pauli operator (other than the identity) that commutes with all the stabilizer generators. The stabilizer formalism also permits a nice characterization of the QEC criteria to be given, namely, $\{E_a\}$ is a set of correctable errors if $E^\dagger_a E_b \not \in \cN(\cS) -\cS$ for all possible error pairs \cite{Knill1997,QECBook}.

\begin{remark}
\label{rem1}
{\em The 3-bit code of Example \ref{ex:rep1} may be easily described within the stabilizer formalism, by letting $\cS$ to be generated, for example, by the two stabilizer operators $S_1= ZZI$ and $S_2=IZZ$  \cite{Viola2001}. The operator $\overline{X}=XXX$ then acts like an encoded $X$ operation on $\cH_\cC$, whereas $\overline{Z}=ZII$ acts like an encoded $Z$ operation, and the repetition code is $[[3,1,1]]$. The two QEC codes independently discovered by Shor \cite{ShorCode} and Steane \cite{SteaneCode} for correcting arbitrary independent single-qubit errors are also both stabilizer codes, corresponding to parameters $[[9,1,3]]$ and $[[7,1,3]]$, respectively; likewise, the smallest, $5$-bit code that also correct for arbitrary single-qubit errors  \cite{PerfectCode} corresponds to $[[5,1,3]]$, as we will also further discuss in Sec. \ref{sub:examples}. } 
\end{remark}

\subsubsection{Subsystem codes}
\label{sub:sc}

The stabilizer formalism has also been extended to codes of the general subsystem form and operator QEC \cite{Poulin2005}. This is done by first specifying a $2^{n-r}$-dimensional subspace $\cH_\cS$ as described in the previous subsection, namely, by specifying $r$ commuting Pauli operators, which generate a stabilizer subgroup $\cS$. We then find  the Pauli subgroup associated to the centralizer of the stabilizer group, and in particular a canonical set of the centralizer generators: As before, these include $iI$ and $\cS,$  and $n-r$ virtual qubits identified by the $\hat X_\ell,\hat Z_\ell$ pairs. In this case, however, only $s$ of these qubits are assigned to encode logical information, while the remaining $n-r-s$ are gauge qubits. Recall that identification of the virtual qubits in the stabilized subspace is, in general, highly non-unique: In some cases, this freedom can be used to identify the gauge qubits so that $\cH_\cS \simeq \cH_L \otimes \cH_F$, with the errors affecting the physical system, restricted to $\cH_\cS$, acting non-trivially {\em only} on $\cH_F$. This implies that one has to employ QEC only to maintain the information inside $\cH_\cS$, as the errors on the gauge qubits do not affect the information encoded in $\cH_L$.

If the stabilizer subsystem code $\cC =\cD(\cH_L)$ encodes $k$ logical qubits into $n$ physical qubits with distance $d$, using $g$ gauge qubits, it is said to be a $[[n,k,g,d]]$ code \cite{Poulin2005,QECBook}. Similarly to the general case, one can obtain stabilizer subspace codes from subsystem codes, and vice versa. In fact:

(1) Every subsystem code can be turned into a standard (subspace) stabilizer code by extending the stabilizer group with extra $n-r-s$ commuting operators that act on the gauge qubits. These impose that the gauge qubits be in a specified pure state, and transform the $[[n,k,g,d]]$ subsystem code into a standard $[[n, k, d]]$ stabilizer code. In doing this, there is a price to pay: The resulting subspace code must now be actively correcting for a new set of error operators, determined by how we choose the additional stabilizers, whereas in the original subsystem code, the same errors did not need to be corrected, since they acted on the gauge qubits.

(2) Reversing the procedure, every subspace stabilizer code, say a $[[n, k, d]]$ code, can be viewed as a subsystem $[[n,k,g,d]]$ code for {\em some} number $g \geq 0$ of gauge qubits, where the latter have been prepared in a pure state. However, in this case the choice of the gauge cannot be arbitrary, or the ability of the code to tolerate certain errors may be lost. It is not easy to determine the maximum $g$ for a given subspace code, though upper bounds exist \cite{Poulin2005}. For example, to the best of our knowledge neither the 5-bit code nor Steane's 7-bit code allow for a subsystem representation with non-trivial gauge qubit, that is, for these codes $g=0$. For Shor's $[[9,1,3]]$ code instead, Poulin provides in \cite{Poulin2005} a $[[9,1,3,3]]$ version of the code with 3 gauge qubits suggesting it is the maximal one, but a $[[9,1,4,3]]$ version with 4 gauge qubits was also later constructed \cite{Breuckmann-thesis}. 

Two main advantages of subsystem codes are that the measurements needed to extract syndrome information are, in general, sparser and the errors only need to be corrected modulo gauge freedom, which may improve fault-tolerance thresholds \cite{Aliferis,Bacon}. More directly relevant to our discussion, subsystem codes may also allow for a certain degree of robustness in the encoding procedure \cite{Klappenecker2009} -- as we shall expand upon after properly defining the encoding task in the next section.

\section{Dissipative quantum encoders}

\subsection{The encoding task}
\label{sec:dqe} 

We now specifically focus on the problem of encoding information in the correctable subsystem. We shall assume that the quantum information (i.e., a quantum state) of interest is initially stored in an {\em upload subsystem} that is isomorphic to the code subsystem and easy to prepare and manipulate, yet unprotected from noise. Typically, such a subsystem will be either directly identifiable with some physical qubits, or emerging from system-specific symmetries and control capabilities \cite{Knill2000,Viola2001}. 

\begin{ex}[Encoding the repetition code]
\label{ex:rep1-enc}
{\em In the repetition code example, it is natural to consider one of the physical qubits as the upload subsystem: Without loss of generality, we choose the first physical qubit.  Then, a simple (unitary) option for translating the initial information into a codeword is offered by any unitary $U_P$ such that
\[U_P\ket{x}\otimes\ket{00}=\ket{xxx}, \qquad x\in\{0,1\} .\]
In the circuit model of quantum computation, the ``encoder'' $U_P$ is realized through a unitary circuit involving single- and two-qubit gates (e.g., it may be obtained from a sequence of CNOT gates \cite{Nielsen-Chuang}). Importantly, this unitary encoding requires the physical qubits that are not used as upload qubits to be prepared in the pure state $\ket{00}.$ This requirement can be relaxed in two ways: 
(1) If we consider the subsystem version of the repetition code, the state of the factor state $\tau_F$ does not affect the encoded information. Hence, we can use as encoder {\em any} unitary such that
\[\tilde U_P\ket{x}\otimes\ket{\phi}=U_L\ket{x}\otimes\ket{\phi'},\]
where $U_L$ is the unitary transformation given in Eq. \eqref{rep} and $\ket{\phi},\ket{\phi'}$ any pair of pure states on two qubits. This implies that $\tilde U_P$ is of the form $U_L (I_1\otimes U_{23}),$ and that the initial state of the second and third qubits is irrelevant to the encoding.  
(2) Otherwise, we can include a (necessarily dissipative) initialization step in the encoding protocol, leading one to consider the CPTP map:
\[\Phi_P ( \rho\otimes\tau_F)=U_P \left[ \,\tr_{2,3}(\rho\otimes\tau_F)\otimes \ket{00}\bra{00} \,\right] U_P^\dag .\]
The latter achieves the correct encoding irrespective of the initialization factor state $\tau_F$ (and, in fact, it correctly encodes the reduced state of the first qubit in $\cC$ even if the input state is not factorized). This illustrates how dissipative dynamics can, in principle, provide additional robustness with respect to errors in the initialization phase of an encoding procedure. Systematic ways to follow the unitary approach (1) with quantum circuits for stabilizer subsystem codes have been proposed in \cite{Klappenecker2009}. Instead, we shall focus on the dissipation-based approach (2).}
\end{ex}

Inspired by the above example, and in line with typical implementations of encoding protocols, 
we consider the task of {\em encoding information} in a quantum code $\mathcal{C}=\Phi(\mathcal{D}(\Hi_L))$ associated to a 1-isometry $\Phi$ and a general subsystem decomposition $\cH_{S}\otimes \cH_{F}\oplus\cH_{R},$ as entailing two steps:
\medskip

{\bf Step 1: Logical encoding (initialization).} The abstract information to be encoded, a density 
operator $\rho\in\cD(\Hi_L),$ is first ``uploaded'' in a physical subsystem that is easy to manipulate and initialize in the desired state, but offers no protection against noise. In full generality, this is done by identifying a subsystem $\Hi_{S'}$ of $\cH_{P}=\cH_{S'}\otimes \cH_{F'}\oplus\cH_{R'},$ with $\Hi_{S'}\simeq\Hi_{L}.$  After this initial information upload, the state of the physical system is some density operator
\begin{equation} 
\Phi_L(\rho)=\rho_{P'}\simeq \rho\otimes \tau_{F'} \oplus 0_{R'}, 
\label{enc1}
\eeq
where $\tau_{F'}\in\cD(\Hi_{F'}).$ The map $\Phi_L$ between the abstract state $\rho$ to be encoded and the initialized state $\rho_{P'}$ must be a $1$-isometric CPTP embedding in order for the information to be retrievable. Note that $\Hi_{F'}$ in the initialization subsystem does not need to be isomorphic to $\Hi_F$ of the code.

\medskip 

{\bf Step 2: Physical encoding.} A CPTP evolution on the physical degrees of freedom, $\Phi_P$, transfers the initialized state to the 
corresponding encoded state in the code according to
\beq 
\Phi_P(\rho_{P'})=\overline\rho\simeq\rho\otimes \tau_F \oplus 0_R,
\label{enc2}
\eeq
where the last state is now correctly initialized in ${\cal C}.$ This $\Phi_P$ can be either obtained unitarily (via a quantum circuit) or dissipatively.  The full {\em encoding protocol} is associated to the concatenation $\Phi =\Phi_P\circ \Phi_L$.

\medskip

In this paper, we shall assume that a nominal $\Phi_L$ is given, and focus on the task of {designing the physical encoder $\Phi_P,$} and its robustness with respect to the initialization. At its core, this task requires that $ \Phi_P\circ\Phi_L(\rho)=\overline\rho.$ Explicitly, in terms of the subsystem decompositions in Eqs. \eqref{enc1}-\eqref{enc2} corresponding to the initialized information and the code, the latter reads: 
\begin{equation}
\Phi_P(\rho\otimes\tau_{F'}\oplus 0_{R'})=\overline\rho\simeq\rho\otimes\tau_F\oplus 0_{R}, 
\label{enc3}
\end{equation}
that is, the physical encoder $\Phi_P$ must map the information uploaded in the initialization step to the correct corresponding state in the code.  

On the one hand, a key difference between unitary and dissipative dynamics for this problem is that the latter allow for contractive, irreversible evolution, such that more than one input state can be mapped to the same target. On the other hand, unitary evolution preserves distinguishability. Formally, we are led to a definition of dissipative encoding that highlights this feature by introducing the basin of attraction of each code state.

\begin{defin}[Dissipative encoder] A (physical) \emph{dissipative encoder} for a code ${\cal C}$ is a (non-unitary) CPTP map $\Phi_P:\mathcal{D}(\Hi_P) \rightarrow\mathcal{B}(\Hi_P)$ as in Eq. \eqref{enc2}, such that for every encoded state $\overline\rho\in{\cal C}, $ the initialized physical state in Eq. \eqref{enc1} obeys $\rho_{P'} \in {\cal B}_{\overline\rho},$ where the {\em basin of attraction} ${\cal B}_{\overline\rho}=\Phi_P^{-1}(\overline\rho)$ is the pre-image of $\overline\rho.$ 
\end {defin}
Robustness in the encoding thus corresponds to having basins of attraction that contain more than just $\rho_{P'}.$ The advantage of using dissipative encoders is essentially related to having non-trivial ${\cal B}_{\overline\rho}$: These allow for tolerance with respect to (certain) errors in the initialization of the upload subsystems, as well as freedom in the design of the map $\Phi_L,$ thus offering potentially easier initialization procedures.

\begin{remark}[Unitary encoders can exhibit robustness only for subsystem codes] {\em Let the map $\Phi_P(\cdot)=U_P\cdot U^{\dagger}_P$ be a unitary physical encoder. As in the repetition code example, the to-be-encoded quantum information $\ket{\psi}$ is initialized in the upload subsystem, while the remaining part of the system is initialized in some fixed pure state $\ket{\phi}$. Then, $\Phi_P$ is a global unitary transformation designed to map $\ket{\psi}\otimes\ket{\phi}$ into the encoded state $\ket{\overline{\psi}}$. In general, the success of the encoding requires the remaining part of the system to be sufficiently well-prepared in $\ket{\phi}$. As we have seen for the repetition code, the unitary encoding is not generally robust to errors in the initialization of $\ket{\phi}$. For example, if the remaining system is afflicted by an error $E$ which transforms $\ket{\phi}$ to an orthogonal state, $\bra{\phi}E\ket{\phi}=\braket{\phi}{\phi'}=0$, the subsequent encoded state is sure to be orthogonal to the intended encoded state. However, a unitary encoder {\em can} be robust with respect to the initialization of some of the subsystems that are not the upload ones, corresponding to the gauge degrees of freedom of the code. The encoded state is then understood within the subsystem code/operator QEC framework, where the gauge degrees of freedom can be mixed -- as opposed to subspace codes that, as we recalled, must correspond to pure gauge states. On the other hand, if a subspace code does not admit a subsystem decomposition with some non-trivial gauge degrees of freedom, the only option to allow for non-pure states in the upload subsystem is to use dissipative encoders at the physical encoding stage.}
\end{remark}

\subsection{Dissipative encoders and their basin of attraction}
\label{sub:basin}

In view of the above discussion, it becomes important to investigate the general structure of the basin of attraction for a subsystem code, and understand how such a basin can be made as large as possible -- at least when no other constraints are in place on the allowed dynamics. 

\subsubsection{Faulty initializations} 

Assume that in the initialization step, the physical system is expected to be mapped in a state $\rho_{P'}$, according to the subsystem structure $\cH_{P}\simeq \cH_{L'}\otimes \cH_{F'}\oplus\cH_{R'}$, with $\dim(\cH_{L'})=\dim(\cH_{L})$, via a given {\em nominal} $\Phi_L$. To explore the structure of the potential basin of attraction, it is convenient to consider a second CPTP map, $\tilde\Phi_{L},$ which represents a {\em faulty initialization} and maps the logical information in a different state $\rho_{P''}$. We first determine what properties such an initialization must have in order for a physical encoder $\Phi_P,$ originally designed to encode the outputs of $\Phi_L,$ to be able to also map the faulty initialized $\rho_{P''}$ of $\tilde\Phi_{L}$ onto the correct encoded state $\bar{\rho} \in \cC$.

As already remarked, it follows from the results in \cite{Ticozzi2010} that $\cal C$ must be a $1$-isometric embedding of the logical information in $\cH_P$ via an isometry $\Phi$. Then it is immediate to see that a dissipative encoder $\Phi_P$ for the faulty initialization exists, meaning it satisfies $\Phi_P\circ\tilde\Phi_{L}=\Phi,$ if and only if $\tilde\Phi_{L}$ is itself a trace-norm isometry. However, this is not sufficient to our aim, since we wish to have the {\em same} $\Phi_P$ to be able to encode the output of $\Phi_L$ as well. We thus give the following:
\begin{defin}
\label{def:ft}
We say that a dissipative encoder $\Phi_P,$  such that $\Phi_P\circ\Phi_L(\rho)=\overline \rho,$ {\em tolerates the faulty initialization} $\tilde\Phi_L$ if
\beq
\label{eq:tolerate-pre}
\Phi_P\circ\tilde\Phi_{L}(\rho)=\Phi_P\circ\Phi_{L}(\rho),\quad \forall\rho\in\cD(\cH_{L}). 
\eeq
\end{defin}

\smallskip

\noindent 
Note that the physical encoder is required to be the same on the left and the right of Eq. \eqref{eq:tolerate-pre}. In what follows, it will be useful to consider a suitable decomposition of the faulty initialization. Namely, we can think that $\tilde \Phi_L$ is the concatenation of the nominal map $\Phi_L$, followed by a noise map ${\cal N}.$ Assuming this structure of $\tilde \Phi_L$ does not imply a loss of generality. In fact, both $\tilde \Phi_L$ and $\Phi_L$ are isometric embedding, otherwise one would have degradation of distinguishability and impossibility of exact decoding. Hence, the image of both maps must correspond to a subsystem representation of the logical degrees of freedom. Any such subsystem representation can be easily obtained from another one by using a map ${\cal N}$ which maps the states initialized in one subsystem in the corresponding ones of the other, and it is suitably completed to be CPTP. If we are considering more than one possible faulty initialization, we can then think that the corresponding $\Ni$ corresponds to a certain family of ``initialization errors'', which act with a certain probability. In this case, we demand that $\Phi_P$ correctly encodes {\em both} the intended $\Phi_L$ and the faulty initializations $\tilde\Phi_L=\Ni\circ\Phi_L$. 
\begin{defin}
\label{def:ft-noise}
We say that a dissipative encoder $\Phi_P$, such that $\Phi_P\circ\Phi_L(\rho)=\overline \rho,$ $\forall\rho\in\cD(\cH_{L}),$ {\em tolerates the noise map} $\cN$ if
\beq
\label{eq:tolerate}
\Phi_P\circ\cN\circ\Phi_L(\rho)=\Phi_P\circ\Phi_{L}(\rho),\quad \forall\rho\in\cD(\cH_{L}).
\eeq
\end{defin}

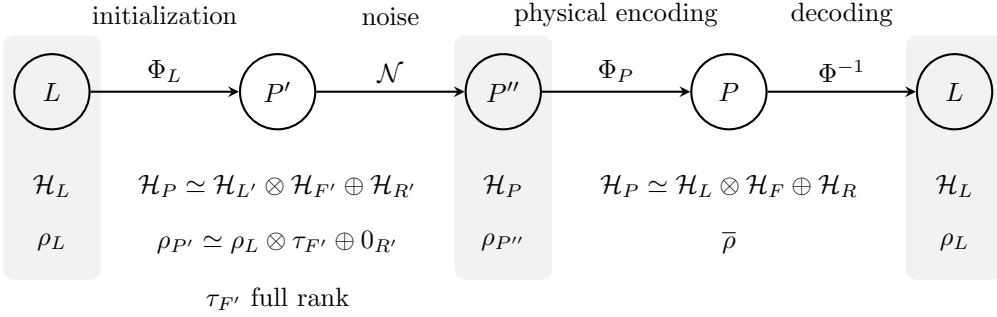
\begin{figure}[t]
\label{figurescheme}
\begin{tikzpicture}
\draw[thick,fill=gray!10,draw=none,rounded corners] (-0.65,0.75) -- (-0.65,-2.5) -- (0.65,-2.5) -- (0.65,0.75) -- cycle;
\draw[thick,fill=gray!10,draw=none,rounded corners] (5.35,0.75) -- (5.35,-2.5) -- (6.65,-2.5) -- (6.65,0.75) -- cycle;
\draw[thick,fill=gray!10,draw=none,rounded corners] (11.35,0.75) -- (11.35,-2.5) -- (12.65,-2.5) -- (12.65,0.75) -- cycle;
\node[thick,circle,draw,minimum size=1cm] (A) at (0,0) {$L$};
\node[thick,circle,draw,minimum size=1cm] (B) at (3,0) {$P'$};
\node[thick,circle,draw,minimum size=1cm] (C) at (6,0) {$P''$};
\node[thick,circle,draw,minimum size=1cm] (D) at (9,0) {$P$};
\node[thick,circle,draw,minimum size=1cm] (E) at (12,0) {$L$};
\node[minimum size=1cm,below of=A,node distance=1.25cm]  (Ab) {$\cH_{L}$};
\node[minimum size=1cm,below of=B,node distance=1.25cm]  (Bb) {$\cH_P\simeq\cH_{L'}\otimes \cH_{F'}\oplus\cH_{R'}$};
\node[minimum size=1cm,below of=C,node distance=1.25cm]  (Cb) {$\cH_P$};
\node[minimum size=1cm,below of=D,node distance=1.25cm]  (Db) {$\cH_{P}\simeq\cH_{L}\otimes \cH_{F}\oplus\cH_{R}$};
\node[minimum size=1cm,below of=E,node distance=1.25cm]  (Eb) {$\cH_{L}$};
\node[minimum size=1cm,below of=Ab,node distance=0.75cm]  (Ac) {$\rho_{L}$};
\node[minimum size=1cm,below of=Bb,node distance=0.75cm]  (Bc) {$\rho_{P'}\simeq \rho_{L}\otimes\tau_{F'}\oplus 0_{R'}$};
\node[minimum size=1cm,below of=Bc,node distance=0.75cm]  (Bc) {$\tau_{F'}$ full rank};
\node[minimum size=1cm,below of=Cb,node distance=0.75cm]  (Cc) {$\rho_{P''}$};
\node[minimum size=1cm,below of=Db,node distance=0.75cm]  (Dc) {$\overline\rho$};
\node[minimum size=1cm,below of=Eb,node distance=0.75cm]  (Ec) {$\rho_{L}$};
\draw[thick,-stealth] (A) -- node[above] (phi) {$\Phi_L$}  (B);
\draw[thick,-stealth] (B) -- node[above] (N) {$\cN$} (C);
\draw[thick,-stealth] (C) -- node[above] (Enc) {$\Phi_P$} (D);
\draw[thick,-stealth] (D) -- node[above] (D) {$\Phi^{-1}$} (E);
\node[minimum size=1cm,above of=phi,node distance=0.75cm]  {\small initialization};
\node[minimum size=1cm,above of=N,node distance=0.75cm]  {\small noise};
\node[minimum size=1cm,above of=Enc,node distance=0.75cm]  {\small physical encoding};
\node[minimum size=1cm,above of=D,node distance=0.75cm] {\small decoding};
\end{tikzpicture}
\vspace*{-3mm}
\caption{Schematic picture of an encoding protocol under faulty initialization. The logical state $\rho_{L}\in\cD(\cH_{L})$ is first initialized in the physical subsystem $\cH_{P}\simeq \cH_{L'}\otimes \cH_{F'}\oplus\cH_{R'}$, with $\dim(\cH_{L'})=\dim(\cH_{L})$, through the logical encoder $\Phi_L$. In a second stage, the initialized state $\rho_{P'}\in\cD(\cH_{P})$ is corrupted by the action of a noise map $\cN$, yielding the ``noisy'' embedded states $\rho_{P''}\in\cD(\cH_{P})$. Then, $\Phi_P$ maps $\rho_{P''}$ into an encoded state $\overline\rho$. Finally, decoding extracts the information encoded in the state $\overline\rho$. }
\label{faulty}
\end{figure}

With reference to the above framework, also depicted in Fig. \ref{faulty}, we shall investigate what properties $\tilde \Phi_L$, or equivalently the noise maps $\cN,$ must possess in order to be tolerated by an encoding map $\Phi_P$. As a first result, it is straightforward to establish the following:

\begin{lemma}
\label{lem:convex}
If the encoding map $\Phi_P$ tolerates the noise action $\cN$, then it tolerates any convex combination of the noise and the identity operator, i.e. $\cM_{\lambda}\equiv \lambda\cN+(1-\lambda)\cI$, with $\lambda\in[0,1]$.
\end{lemma}
\begin{proof}
The proof follows by linearity of $\Phi_P$. In fact, since $\Phi_P$ is a linear map correcting $\cN$, for every $\rho_{L}\in\cD(\cH_{L})$ we have: 
\begin{align*}
\hspace*{1cm}\Phi_P\circ\cM_{\lambda}\circ\Phi_L(\rho_{L})&=\Phi_P\circ(\lambda\cN+(1-\lambda)\cI)\circ\Phi_L(\rho_{L}) \\ &= \lambda\Phi_P\circ\cN\circ\Phi_L(\rho_{L})+(1-\lambda)\Phi_P\circ\Phi_L(\rho_{L}) =\Phi_P\circ\Phi_L(\rho_{L}). \hspace*{18mm}\Box
\end{align*}
\end{proof}

\vspace*{-6mm}

\subsubsection{Compatible subsystems}

Next, we wish to investigate in more detail what faulty initializations can be tolerated by some $\Phi_P$. Since all concatenation of maps must be trace-norm isometries on their inputs in order to faithfully preserve the information contained in $\rho$, the image of $\Phi_L$ and ${\cal N}\circ\Phi_L$ must have the structure of a general subsystem code. Let 
\begin{align*}
\rho_{P'}\simeq\rho_{L}\otimes \tau_{F'}\oplus 0_{R'},\quad \rho_{P''}\simeq\rho_{L}\otimes \tau_{F''}\oplus 0_{R''},
\end{align*} where $\tau_{F'}$ and $\tau_{F''}$ are of full rank. These denote the states associated to the following subsystem decompositions of $\cH_{P}$, respectively: 
\begin{align}
\label{eq:sscomp}
 \cH_{P}\simeq\cH_{L'}\otimes \cH_{F'}\oplus \cH_{R'},\quad \cH_{P}\simeq\cH_{L''}\otimes \cH_{F''}\oplus \cH_{R''}, 
\end{align}
with $\dim(\cH_{L'})=\dim(\cH_{L''})=\dim(\cH_{L})$. In order for a faulty initialization to be tolerated by $\Phi_P$, it needs to be in some sense ``compatible'' with the properly initialized information. That is, if we project the faulty initialized states back on the support of the nominal ones, they should exhibit a tensor structure that is of the same form. This notion is made precise in the following definition.
\begin{defin}[Compatible initializations]
\label{def:comp}
Consider two subsystem decomposition as in Eq. \eqref{eq:sscomp}. 
We say that the initializations in $\cH_{L'}$ with co-factor state $\tau_{F'}$, and in $\cH_{L''}$ with co-factor state $\tau_{F''}$, are {\em compatible} if
\begin{align}
\label{eq:compsub}
	\begin{cases}
	\Pi_{L'F'}\rho_{P''}\Pi_{L'F'}\simeq\rho_{L}\otimes\tilde\tau_{F'}\oplus 0_{R'},\\
	\Pi_{L''F''}\rho_{P'}\Pi_{L''F''}\simeq\rho_{L}\otimes\tilde\tau_{F''}\oplus 0_{R''},
	\end{cases}\quad \forall \rho_{L}\in\cD(\cH_{L}),
\end{align}
where $\Pi_{L'F'}$, $\Pi_{L''F''}$ are the orthogonal projections onto the subspaces $\cH_{L'}\otimes \cH_{F'}$, $\cH_{L''}\otimes \cH_{F''}$, respectively, and $\tilde{\tau}_{F'},\, \tilde{\tau}_{F''}\geq 0$.
\end{defin}
\begin{remark} {\em In principle, one may allow for the projected states, that is, the right-hand sides of Eq. \eqref{eq:compsub}, to have a co-factor $\tilde\tau_{F',F''}$ that depends on the encoded state $\rho_L$. In Appendix \ref{app:cofactor} we explicitly show that this seemingly looser requirements is actually equivalent to the definition of compatible initializations given above.}
\end{remark}

Let us now consider the decomposition $\cH_{P}\simeq\cH_{L''}\otimes \cH_{F'',\lambda}\oplus \cH_{R'',\lambda}$  induced by  the isometric embedding $\cM_{\lambda}\circ \Phi_L,$ with $\cM_{\lambda}\equiv \lambda\cN+(1-\lambda)\cI$, $\lambda\in[0,1]$. We also denote by $\rho_{P'',\lambda}\equiv \cM_{\lambda}(\rho_{P'})=\rho_{L}\otimes \tau_{F'',\lambda}\oplus 0_{R'',\lambda}$ the initialized state in $\cH_{L''}$ with co-factor state $\tau_{F'',\lambda}\in \cD(\cH_{F'',\lambda})$ of full rank. With these definitions in place, we are now in a position to prove the main theorem of this section, which clarifies the role of compatible initializations: 

\begin{thm}[Noise-tolerant encoders]
\label{thm:noise}
There exists an encoding map $\Phi_P$ that tolerates the noise action $\cN$ if and only if the two initializations $\rho_{P'}$ and $\rho_{P'',\lambda}$ in the subsystem decompositions $\cH_{L'}\otimes \cH_{F'}\oplus \cH_{R'}$ and $\cH_{L''}\otimes \cH_{F'',\lambda}\oplus \cH_{R'',\lambda}$, respectively, are compatible for all $\lambda\in[0,1]$.
\end{thm}

The proof is given in Appendix \ref{app:main}. In the light of this characterization, we can show that, in general, it is not possible to define of a {\em unique}, ``maximal'' basin of attraction. First, notice that if an initialization is compatible with the nominal one, we can define a subsystem decomposition that comprises both. Intuitively, either one is already a particular case of the other, with $\tau_F$ having support only on a proper subspace of $\Hi_F$, or it is sufficient to augment the dimensionality of the tensor factor $\Hi_F$ by identifying isomorphic copies of $\Hi_L$ in $\Hi_R.$ 

Furthermore, it is easy to see that compatibility is not a transitive property for subsystem decompositions.  In fact, it is possible to construct a counter-example where two faulty initializations are both compatible with the nominal one on its own support, but have mutually incompatible structure on the orthogonal complement. Hence it is not enough to consider all the noise actions compatible with the correct initialization, and construct an overarching, maximal subsystem structure. Nonetheless, it is possible to identify the {\em maximal size} of the gauge subsystem of a tolerable initialization, which corresponds to the integer part of $\dim(\Hi_R)/\dim(\Hi_L).$

If $\dim(\Hi_R)/\dim(\Hi_L)$ is indeed integer, constructing a maximal-dimension gauge subsystem leads to a pure tensor-factor decomposition of $\cH_P$, that is, one with no summand $\cH_R$: 
\[\Hi_P\simeq\Hi_L\otimes\Hi_{F,\max}.\]
In such a case, if the dissipative encoder $\Phi_P$ tolerates faulty encodings in $\Hi_P\simeq\Hi_L\otimes\Hi_{F,\max},$ we find the  basin of attraction for each state to be simply ${\cal B}_{\overline\rho}=\{\rho\otimes\tau_F\ |\ \tau_F\in{\cal D}(\Hi_{F,\max})\}.$

It is worth remarking that these results are existential in nature, aimed to characterize what type of robustness may be attained, in principle, by using dissipative quantum encoders. What can be done in specific scenarios, including under further design constraints, may significantly vary. For example, in Sec. \ref{sec:FTE}, we will develop dissipative encoders that additionally guarantee {\em invariance} of the target code -- at the cost of reducing the basin of attraction -- and that rely on the important class of {\em Markovian} dissipation, as we formally introduce next. 

\subsection{Encoding via Markovian dynamics}
\label{sub:Markov}

\subsubsection{Continuous-time encoders}
\label{sub:CDE}

Continuous-time Markovian dynamics are widely employed to model a variety of both naturally occurring and controlled irreversible behavior, in contexts ranging from quantum statistical mechanics and thermodynamics to continuous quantum measurement and quantum reservoir engineering \cite{AlickiLendi,Kosloff,altafini-introduction}. Their convergence to equilibria is provably always asymptotic \cite{johnson-FTS}. In particular, a number of approaches have been devised for analyzing and constructing continuous-time quantum dynamical semigroups (QDSs) able to ensure stabilization of desired states, subspaces, and subsystems, see e.g.~\cite{ticozzi-QDS,ticozzi-markovian} and references therein. 

The task of designing a physical encoder entails a related, yet more articulated, set of requirements. With reference to Eq. \eqref{enc3}, we say that a generator of a QDS (or ``Liouvillian''), $\Li:\mathcal{B}(\Hi_{P})\rightarrow\mathcal{B}(\Hi_P)$, defines {\em continuous-time dissipative encoding (CDE)} for a subsystem code $\cC$ if for all initialized states $\Phi_L(\rho)=\rho_{P'}\simeq \rho\otimes \tau_{F'} \oplus 0_{R'},$ where $\rho\in\mathcal{D}(\Hi_L)$ and $\tau_{F'}\in \mathcal{D}(\Hi_{F'})$, the evolution converges asymptotically to the intended state in $\cC$, that is, 
\beq
\label{eq:CDD} 
\lim_{t\rightarrow +\infty}e^{\Li t}[\rho_{P'}]=\overline\rho\simeq\rho\otimes \tau_F \oplus 0_R\in\cC.
\eeq
In addition, for multipartite systems one may require the encoding generator to respect some locality constraints. We say that {\em a CDE is quasi-local (QL)} with respect to a specified neighborhood structure $\Ni\equiv \{\Ni_k\}$ if $\Li=\sum_k \Li_{k},$ with $\Li_k$ a generator acting nontrivially only on one neighborhood, that is, a {\em neighborhood map} \cite{johnson-FTS}.

\begin{remark} 
{\em Beside ensuring encoding, Eq. \eqref{eq:CDD} automatically implies that each state encoded in $\cC$ is an invariant (fixed) state for $\Phi_P$, $\Li(\overline\rho)=0.$ Note that the map $\Phi_P$ is formally well-defined {\em only on initialized states}, as the limit of $e^{\Li t}$ exists for initialized input states and their attraction basins. Since all initial states of a QDS converge towards its center manifold (that comprises eigenoperators relative to purely imaginary eigenvalues), non-initialized states could also converge to rotating states, preventing the CPTP map $\Phi_P$ from being well-defined. By requiring the limit to exist, we must have, in particular, eigenoperators corresponding to eigenvalue zero (i.e., non-oscillating) and hence the limit, for each initial condition, is a fixed operator. Also notice that while no explicit robustness against initialization errors is imposed, a CDE can, at least in principle, reabsorb errors {\em asymptotically}. Assume that two different states $\rho_1,\rho_2$ are to be correctly encoded in the same codeword $\rho\in\cC$. Then their difference must converge to zero, which can happen only asymptotically for continuous-time dynamics \cite{johnson-FTS}. Since our discussion will focus on the ability of dissipative encoding to reabsorb errors, this justifies the asymptotic limit in our definition.}
\end{remark}

\subsubsection{Discrete-time encoders}

In scenarios where the physical encoder is implemented via a discrete sequence of operations (unitary and dissipative gates and measurements), a different definition is more appropriate. The discrete-time framework allows for finite-time convergence, and includes more naturally the typical unitary protocols for encoding as a limiting case. 
 
A sequence of CPTP maps  $\{\Ei_k\}$ on $\mathcal{B}(\Hi_{P})$, defines {\em discrete-time dissipative encoding (DDE)} for a subsystem code $\cC$ 
if for all initialized states $\Phi_L(\rho)=\rho_{P'}\simeq\rho\otimes \tau_{F'} \oplus 0_{R'},$ where $\rho\in\mathcal{D}(\Hi_L)$ and $\tau_{F'}\in \mathcal{D}(\Hi_F')$, the evolution converges asymptotically to the intended state in $\cC$, that is, 
\beq 
\lim_{k\rightarrow +\infty}\Ei_k\circ\Ei_{k-1}\circ\ldots\circ\Ei_1[\rho']=\overline\rho\simeq\rho\otimes \tau_F \oplus 0_R\in\cC.
\label{dde}
\eeq

\noindent 
The limit of the concatenated sequence exists by definition, at least for initialized input states, and its extension to a CPTP map $\Phi_P$ is the {discrete-time dissipative physical encoder.} 

In contrast with continuous time, in a discrete-time scenario perfect encoding can be achieved in finite time in principle:
A finite sequence of CPTP maps,  $\{\Ei_k\}_{k=1}^{M}$ on $\mathcal{B}(\Hi_{P})$, defines a {\em finite-time dissipative encoder (FTDE)} for a subsystem code $\cC,$ 
if for all initialized states $\Phi_L(\rho)=\rho_{P'}\simeq \rho\otimes \tau_{F'} \oplus 0_{R'},$ where $\rho\in\mathcal{D}(\Hi_L)$ and $\tau_{F'}\in \mathcal{D}(\Hi_{F'})$, the evolution converges to the intended state in the code in a finite number of steps, that is, 
\begin{equation}
\Ei_M\circ\Ei_{M-1}\circ\ldots\circ\Ei_1[\rho']=\overline\rho\simeq\rho\otimes \tau_F \oplus 0_R\in\cC.
\label{ftde}
\end{equation}
We say that {\em a discrete-time or FT DE is QL} with respect to a neighborhood structure $\Ni$ if each $\Ei_{k}$ is a neighborhood map.

An FTDE allows for exact encoding in finite time, as typical unitary encoders do, while retaining the ability of absorb initialization errors. While in the discrete-time scenario the invariance of the code states is not strictly required to reabsorb errors, imposing invariance allows us to better compare with CDE, and makes the encoding task compatible with QEC protocols, as we shall see in the next section. We will say that a discrete-time DE or FT DE is {\em code preserving} if each state $\overline\rho\in\cC$ satisfies $\Ei_k(\overline\rho)= \overline\rho$, for each of the maps $\Ei_k$ in Eq. \eqref{dde} or Eq. \eqref{ftde}, respectively. 
\section{Finite-time dissipative encoders for stabilizer quantum codes}
\label{sec:FTE}

\subsection{A finite-time encoder for the repetition code}

In order to gain intuition into the general case, we first reconsider the 3-qubit repetition code, within the stabilizer formalism. As mentioned in Remark \ref{rem1}, $\cH_\cC=\textup{span}\{\ket{000},\ket{111}\}=\cH_\cS$ can be associated to stabilizer generators $\{S_1=ZZI, S_2=IZZ\}\in \cS$. As before, we consider the first physical qubit to be the upload qubit, which we assume to be initialized in the logical state $\rho_L$ to be encoded, and we choose the logical operators to be $\overline{X}=XXX$ and $\overline{Z}=ZII$. We now show that $\cC$ may be encoded in finite time from a localized upload qubit using a sequence of two-body CPTP maps, that is, \(\Phi_P=\Phi_{23}\circ\Phi_{12},\) where indexes are understood to label physical qubits~\footnote{One may also construct $\Phi_P$ through maps that act on qubits in the opposite order or, as we will see in the general case, independently from the order altogether.}.

Our strategy is to choose encoding maps that resemble the error-correcting operations of a stabilizer code. As a first step, we propose a structure for the CPTP maps that guarantees that the image of $\Phi_P =\Phi_{23}\circ\Phi_{12}$ corresponds to $\cC$. Each of the two-body maps performs a measurement of the stabilizer generator $S_k$, associated to projectors $ \frac{1}{2}(\identity\pm S_{k}),$ followed by a unitary (Pauli) correction operation $C_k$ in case the outcome corresponding to $\frac{1}{2}(\identity - S_{k})$ is observed. As in stabilizer QEC, we choose $C_k$ so that it anticommutes with $S_k$. Hence, $\frac{1}{2}C_k(\identity - S_{k})=\frac{1}{2}(\identity + S_{k})C_k.$ The Kraus operators of the composed map $\Phi_P$ are then
$$K^{(i_1, i_2)}=\frac{1}{2}(\identity+S_2)C_2^{i_2}\frac{1}{2}(\identity+S_1)C_1^{i_1},\qquad i_1,i_2 \in\{0,1\}.$$ 
In order for the range of these composed maps to be in $\cC$, it is {\em necessary and sufficient} that $[C_2, S_1]=0$. For necessity, consider that $[C_2, S_1]\not =0$. Then, it must be that $\{C_2,S_1\}=0$, since elements in $\cP_n$ either commute or anti-commute. The Kraus operator $K^{(i_1,1)}$ would then be of the form $\frac{1}{2}(\identity+S_2)\frac{1}{2}(\identity-S_1)C_2C_1^{i_1}$, having a range which is orthogonal to $\cH_\cC$. That commutativity is also sufficient follows from the fact that the range of the composed map's Kraus operators is the range of $\frac{1}{4}(\identity+S_2)(\identity+S_1)$, which is equal to the code support (see Proposition \ref{prop:stab} below). 
 
We can narrow down our search further by requiring that the logical operators of the code be left invariant, making the encoder code-preserving. This is equivalent to having the correction operators commute with the logical operators, leaving, in our specific case, a choice of {$C_{1}\in\{IXI, ZYI\}$ and $C_{2}\in \{IIX,IZY\}$.}

Furthermore, for a general correction map, correction operators are only ever defined up to multiplication by the corresponding stabilizer. For instance, the Kraus operators of the first map, including the correction, can be written equivalently with respect to either choices of the correction operator $C_1$, since
$$ \frac{1}{2}(III + ZZI)ZYI=\frac{1}{2}(ZZI + III)(ZZI)ZYI=\frac{1}{2}(III + ZZI)IXI.$$
Therefore, our requirements have effectively singled out one possibility for the encoding maps. The key properties of these maps is that the range of their composition is in $\cC$ and the logical operator values are left invariant. 

Having specified the form of $\Phi_P$, we determine the basin of attraction ${\cal B}_{\overline\rho}$ which ensures that $\Phi_P(\rho_L\otimes\sigma)=\overline{\rho}.$ In this case, since the summand $\cH_R$ is empty, it is sufficient to determine the set of co-factor states $\sigma$ that guarantee the correct encoding above. For this to be achieved, the expectation of the logical operators computed with the output density matrix must coincide with those of the upload qubit, that is, 
\begin{equation}
\label{eq:encodingrelation}
\tr{}(\overline{X}^i\overline{Z}^j\Phi_P(\rho_L\otimes\sigma))=\tr({X^iZ^j\rho_L}), \qquad \forall i,j\in\{0,1\}.
\end{equation}
Since, by construction, $\Phi_P$ leaves the values of the logical operators invariant, the above equation simplifies to 
\( \tr{}{(X^iZ^j\rho_L)}=\tr{}{(\overline{X}^i\overline{Z}^j(\rho_L\otimes\sigma))}. \)
Evaluating this for the repetition code, we obtain
\begin{equation*}
\tr{}{(X^iZ^j\rho_L})=\tr{}{((X^iZ^j\rho_L)\otimes[(XX)^i\sigma])},
\end{equation*}
giving $\tr{}{((XX)^i\sigma)}=1$. Accordingly, the basin of attraction for $\overline\rho$ is the set of density operators $\rho_L\otimes \sigma$ with support of $\sigma$ contained in the $+1$-eigenspace of $XX$.

\subsection{General structure of finite-time encoders}

Building on the previous example, we now construct a FTDE $\Phi_P$ for a given stabilizer (subspace) code, with stabilizer generators $\{S_{k}\}_{k=1}^{r}\subseteq \cP_{n}$, by considering a composition of $r$ encoding CPTP maps 
of the form
\[ \Phi_{k}(\rho)\equiv A_{+,k}\rho A_{+,k}^{\dagger}+A_{-,k}\rho A_{-,k}^{\dagger},\quad k=1,\dots,r, \]
where
\[ A_{+,k}\equiv \frac{1}{2}(I+S_{k}), \ A_{-,k}\equiv \frac{1}{2}C_{k}(I-S_{k}). \]
Here, $\{C_{k}\}_{k=1}^{r}$ are correction-like operators, that we require to satisfy a number of constraints.
 
\medskip

\noindent 
{\bf E1. The code space must be correctly prepared.} $\{C_{k}\}_{k=1}^{r}$ are Pauli operators such that 
\beq
\label{eq:anticommuting}
\{C_{k},S_{k}\}=0, \quad \forall k, 
\eeq 
and
\beq
\label{eq:commuting}
[C_{k},S_{j}]=0, \quad \forall j\ne k.
\eeq
This implies that $A_{-,k}$ can be rewritten as
\beq
\label{eq:altform}
A_{-,k}=\frac{1}{2}(I+S_{k})C_{k}.
\eeq
The latter form is useful in proving that $\Phi_P$ prepares the code subspace in the Schr\"odinger's picture, that is, in establishing the following:
\begin{prop}
\label{prop:stab}
If the conditions in Eqs. \eqref{eq:anticommuting}-\eqref{eq:commuting} are obeyed, then {\em any} concatenation $\Phi_P$ of the $r$ maps $\Phi_1,\ldots,\Phi_r$ prepares the code subspace, \(\Phi_P(\rho)\in{\cal C}\), and each encoded state $\rho\in{\cal C}$ is invariant.
\end{prop}
\begin{proof} If an operator has support on $\cH_\cC$, then it is in the $+1$-eigenspace of all the $A_{+,k}$ operators and in the kernel of all the $A_{-,k},$ and hence it is preserved by each $\Phi_k.$  Thus, $\cH_\cC$ is invariant and, in particular, all the encoded states are fixed states for $\Phi_k$'s and thus $\Phi_P$. Using Eq. \eqref{eq:altform} to represent all the $A_{-,k}$ operators of the $\Phi_k$ and Eq. \eqref{eq:commuting} to ``push'' all correction operators before (to the right of) the projections,  {\em all} Kraus operators of the concatenated maps $\Phi_P,$ independently of the order of the $\Phi_k$, can be written in the form $\Pi_\cC \bar C$, where \[\Pi_\cC\equiv \prod_{\ell=1,\ldots,r}(I+S_k)/2,\] is the projection on the stabilizer subspace $\cH_\cS=\cH_\cC$, and $\bar C$ an ordered product of a subset of correction operators of the selected $A_{-,k}$. Hence, the output of the concatenated map $\Phi_P$ has support contained in the support of $\cC$. Since all the maps are TP, this implies that the $\Phi_P$ stabilizes $\cH_\cC$ in finite time \cite{ticozzi-discretefeedback,johnson-FTS,ticozzi-alternating}.~\qed\end{proof}

It is worth noticing that: (i) Condition \eqref{eq:commuting} is also necessary for invariance given the structure of the maps we chose, as in the repetition-code example;  (ii) if an ordering for the stabilizer operators is fixed, we can replace Eq. \eqref{eq:commuting} with a weaker requirement, namely, $[C_{k},S_{j}]=0$, $\forall j< k$. The stronger condition \eqref{eq:commuting} will imply that the encoding maps can be applied in any order (namely, they guarantee {\em robust FTDE with respect to the map ordering}). Remarkably, we will prove that finding such operators is always possible for stabilizer codes.

To address the encoding, however, this is not sufficient. We need to impose that the upload qubits are correctly mapped to the corresponding ones in $\cC$. To this aim, it is convenient to focus on the effect of the encoder on the observables -- that is, to move to the Heisenberg picture. 

\medskip

\noindent 
{\bf E2. Encoded operators must be extensions of the upload ones.} Consider a partition of the physical subsystems in upload qubits and {(with some abuse of terminology)} 
gauge qubits for the encoding maps: $\cH_P\simeq \cH^{\text{upload}}\otimes \cH^{\text{gauge}}.$ The upload subsystems initially carry the the information that will be transferred into $\cC$ by the encoder; the gauge qubits are the rest. Let $\{X_{k},Z_{k}\}_{k=1}^{n-r}$ be a canonical choice of logical operators for the input subsystem. We require that: 
\beq
\label{eq:encode}
\overline{X}_{1}^{p_{1}}\overline{Z}_{1}^{q_{1}}\cdots \overline{X}_{n-r}^{p_{n-r}}\overline{Z}_{n-r}^{q_{n-r}}=X_{1}^{p_{1}}Z_{1}^{q_{1}}\cdots X_{n-r}^{p_{n-r}}Z_{n-r}^{q_{n-r}}\otimes R_{p,q},
\eeq
for all $p=(p_{1},\dots,p_{n-r})$, $q=(q_{1},\dots,q_{n-r})$ such that $p_{i},q_{i}\in\{0,1\}$, with $R_{p,q}$ being a Pauli operator acting non-trivially on the co-factor (gauge) subsystem $\cH^{\text{gauge}}.$

\medskip

\noindent 
{\bf E3. Encoded operators must be invariant.} While we proved that all states on the code support are indeed invariant, the operators we choose to represent encoded information have support everywhere, and their invariance is not directly guaranteed by the form of the maps $\Phi_k.$ As in the repetition code example, invariance is guaranteed if  
\beq
\label{eq:invc}
[C_{k},\overline{X}_{j}]=[C_{k},\overline{Z}_{j}]=0, \quad\forall j\ne k.
\eeq
Condition \eqref{eq:invc}, in the dual (Heisenberg) picture, ensures that the logical operator are fixed points for the 
maps $\Phi_k$ and it further allows us to easily determine the basin of attractions of $\Phi_P$. 
 
\medskip

We can now show that the $\Phi_P$ we constructed is indeed a valid encoder for the target code:

\begin{prop}
\label{prop:DE}
The concatenation $\Phi_P=\Phi_{r}\circ\cdots\circ\Phi_{1},$ for any ordering of the $\Phi_{i}$, is a valid FTDE for the stabilizer code ${\cal C}$ if requirements E1.-E2.-E3. above are satisfied. The common $+1$-eigenspace of the $R_{p,q}$ identifies the basin of attraction of co-factor states $\sigma$ for the associated physical encoder, that is, 
\(\Phi_P(\rho_L\otimes\sigma)=\bar\rho\in\cal C,\)
where $\rho_L$ is the state on the first $n-r$ qubits to be encoded, and ${\rm supp}(\sigma)$ is contained in the common +1-eigenspace of $R_{p,q}$.
\end{prop}

\begin{proof}
Proper encoding (in the Heisenberg picture) is ensured if:
\beqa&&\Tr(\overline{X}_{1}^{p_{1}}\overline{Z}_{1}^{q_{1}}\cdots \overline{X}_{n-r}^{p_{n-r}}\overline{Z}_{n-r}^{q_{n-r}}\Phi_P(\rho_L\otimes \sigma))=\Tr(X_{1}^{p_{1}}Z_{1}^{q_{1}}\cdots X_{n-r}^{p_{n-r}}Z_{n-r}^{q_{n-r}}\rho_L),
\label{eq:want}
\eeqa
for all $p=(p_{1},\dots,p_{n-r})$, $q=(q_{1},\dots,q_{n-r})$ such that $p_{i},q_{i}\in\{0,1\}$. Assuming that the conditions in Eqs. \eqref{eq:commuting}, \eqref{eq:encode}, \eqref{eq:invc} hold, we have, for all $p,q$:
\beqa
\nonumber\Tr(\overline{X}_{1}^{p_{1}}\overline{Z}_{1}^{q_{1}}\cdots \overline{X}_{n-r}^{p_{n-r}}\overline{Z}_{n-r}^{q_{n-r}}\Phi_P(\rho_L\otimes \sigma))
&=&\Tr\big(\Phi_P^\dag(\overline{X}_{1}^{p_{1}}\overline{Z}_{1}^{q_{1}}\!\cdots \overline{X}_{n-r}^{p_{n-r}}\overline{Z}_{n-r}^{q_{n-r}})(\rho_L\otimes \sigma)\big)\\
&\nonumber=&\Tr(\overline{X}_{1}^{p_{1}}\overline{Z}_{1}^{q_{1}}\cdots \overline{X}_{n-r}^{p_{n-r}}\overline{Z}_{n-r}^{q_{n-r}}(\rho_L\otimes \sigma))\\
&=&\Tr(X_{1}^{p_{1}}Z_{1}^{q_{1}}\cdots X_{n-r}^{p_{n-r}}Z_{n-r}^{q_{n-r}}\rho_L)\Tr(R_{p,q}\sigma). 
\label{eq:new}
\eeqa
In the last two equations we have used two facts: (i) the encoded operators are invariant for the (dual) encoder $\Phi_P^\dag$, as they commute with both $S_k$ and $C_k$, by Eq. \eqref{eq:invc}. Therefore, for each $k$, 
\[A^\dag_{\pm,k}(\overline{X}_{1}^{p_{1}}\overline{Z}_{1}^{q_{1}}\cdots \overline{X}_{n-r}^{p_{n-r}}\overline{Z}_{n-r}^{q_{n-r}})A_{\pm,k}=(\overline{X}_{1}^{p_{1}}\overline{Z}_{1}^{q_{1}}\cdots \overline{X}_{n-r}^{p_{n-r}}\overline{Z}_{n-r}^{q_{n-r}})A^\dag_{\pm,k}A_{\pm,k},\] from which invariance for the (unital) concatenation map $\Phi_P^\dag$  follows; and (ii) the explicit form in Eq. \eqref{eq:encode} for the encoded operators. Then \eqref{eq:want} is equal to \eqref{eq:new} if and only if is the co-factor state $\sigma$ is such that $\Tr(R_{p,q}\sigma )=1,$ for all $p,q.$  \qed
\end{proof}

\begin{remark}
{\em A similar construction to the one described above can be used to build a CDE, as defined in Sec. \ref{sub:CDE}. More precisely, the semigroup generator \( \Li(\rho)\equiv \Phi_P(\rho)-\rho, \) where $\Phi_P$ is a DDE, defines a CDE from the first $n-r$ qubits to the subspace code ${\cal C},$ with basin of attraction the common $+1$-eigenspace~of~$R_{p,q}$. The full proof, which employs a different approach leveraging Lyapunov techniques, is given in \cite{encoding-CDC}.}
\end{remark}

\subsection{Main result and implications}

In the following theorem, which is the main result of this section, we both establish that, for stabilizer codes, it is always possible to construct encoding maps that satisfy the structural constraints specified above, and provide an explicit construction achieving that.  As a corollary of the analysis, we further show that the basin of attraction is always non-empty, and in fact corresponds to a stabilizer subspace of the co-factor space.
\begin{thm}[FTDE for stabilizer codes] 
\label{thm:stabilizer-enc}
Given a stabilizer group $\cS$ associated to a (subspace) code $\cC$, there exists a set of generators $\{S_i\}_{i=1}^{r}$,  logical operators $\{\overline{X}_i,\overline{Z}_i\}_{i=1}^{n-r}$ and correction maps $\{C_i\}_{i=1}^{r}$ satisfying the encoding requirements E1.-E2.-E3 for FTDE in $\cC$.
\end{thm}
\begin{proof}
The proof is constructive. First, we show how to define a set of logical operators $\{\overline{X}_i,\overline{Z}_i\}_{i=1}^{n-r}$ satisfying condition E3. To this end, we exploit the check matrix representation of the stabilizer generators (see Appendix \ref{sympl}). Specifically, we can express the check matrix of the stabilizer group $\cS$ as
\begin{align}
\label{eq:standard-form}
S = \begin{blockarray}{ccccccc}
 {\scriptscriptstyle n-r} & {\scriptscriptstyle r-k} & {\scriptscriptstyle k} & {\scriptscriptstyle n-r} & {\scriptscriptstyle r-k} & {\scriptscriptstyle k} \\
      \begin{block}{[ccc|ccc]l}
       A_1 & A_2 & I & B & 0 & C & \ {\scriptscriptstyle k} \\ 
           0 & 0 & 0 & D & I & E & \ {\scriptscriptstyle r-k} \\
      \end{block}
      & & &
    \end{blockarray}
    \in \Zbb_{2}^{r\times 2n}
\end{align}
for binary matrices $A_{1}$, $A_{2}$, $B$, $C$, $D$, $E$ of suitable dimensions (specified by the top and right labels). In particular, Eq.~\eqref{eq:standard-form} corresponds to the ``standard form'' presented in  \cite[Sec.~10.5.7]{Nielsen-Chuang}, up to a relabeling of qubits. Next, we define
\begin{align}
G_z&\equiv \begin{blockarray}{ccccccc}
{\scriptscriptstyle n-r} & {\scriptscriptstyle r-k} & {\scriptscriptstyle k} & {\scriptscriptstyle n-r} & {\scriptscriptstyle r-k} & {\scriptscriptstyle k} \\
      \begin{block}{[ccc|ccc]l}
      0 & 0 & 0 & I & 0 & A_1^{\top} & \ {\scriptscriptstyle n-r} \\ 
      \end{block}
      & & &
    \end{blockarray} \in \Zbb_{2}^{(n-r)\times 2n} , 
    \label{eq:Gz} \\
G_x& \equiv \begin{blockarray}{ccccccc}
{\scriptscriptstyle n-r} & {\scriptscriptstyle r-k} & {\scriptscriptstyle k} & {\scriptscriptstyle n-r} & {\scriptscriptstyle r-k} & {\scriptscriptstyle k} \\
      \begin{block}{[ccc|ccc]l}
      I & D^{\top} & 0 & 0 & 0 & B^{\top} & \ {\scriptscriptstyle n-r} \\ 
      \end{block}
      & & &
    \end{blockarray} \in \Zbb_{2}^{(n-r)\times 2n} .
    \label{eq:Gx}
\end{align}
Notice that the rows of $G_z$ and $G_x$ commute with every stabilizer generator in Eq. \eqref{eq:standard-form}, since
\[ S\Lambda G_z^{\top} = 0\ \ \text{ and }\ \ S\Lambda G_x^{\top} = 0. \]
Further, it holds $G_z \Lambda G_x^{\top}=I$. This in turn implies that the Pauli operators defined by the rows of $G_z$ and $G_x$ yield a canonical set of generators. We now define a set of generators of $\cS$, $\{S_i\}_{i=1}^{r}$, via the rows of $S$ in Eq. \eqref{eq:standard-form} and a set of encoded $X$ operators, $\{\overline{X}_i\}_{i=1}^{n-r}$, encoded $Z$ operators, $\{\overline{Z}_i\}_{i=1}^{n-r}$, via the rows of $G_x$, $G_z$, respectively. If we assume, without loss of generality, that the first $r$ qubits represent the physical input (upload) qubits, the latter encoded logical operators are seen to satisfy condition E2.

To complete the proof, we need to find a set of correction operators satisfying E1.-E3. To this end, we first observe that the rows of the matrix
\[ \overline{S}\equiv \left[\begin{array}{c}
      	S\\
	\hline
	G_z \\
	\hline
	G_x
\end{array}\right] \in \Zbb_{2}^{(2n-r)\times 2n} \]
are linearly independent by construction. By using Propositions 10.3-10.4 in \cite{Nielsen-Chuang}, there exists a set of correction operators $\{C_i\}_{i=1}^{r}$ in $\cP_n$ satisfying the desired conditions. More precisely, since the rows of $\overline{S}$ are linearly independent, there always exists a vector $x_{i}\in\Zbb_{2}^{2n}$ such that $\overline{S}\Lambda x_{i}=e_{i}$, where $e_{i}$ denotes the $(2n-r)$-dimensional vector with a 1 in the $i$-th position and $0$s elsewhere. Suppose that $x_{i}^{\top}$ corresponds to the row vector representation of $C_{i}$. In view of condition $\overline{S}\Lambda x_{i}=e_{i}$, $C_{i}$ commutes with all Pauli operators defined via the rows of $\overline{S}$, with the only exception of the $i$-th operator. We have therefore constructed a set $\{C_{i}\}_{i=1}^{r}$ that satisfies conditions E1.-E3. \qed
\end{proof}

Through the construction proposed in the proof, we can explicitly obtain the operators $R_{p,q}$ of Eq. \eqref{eq:encode}. These satisfy some additional interesting properties (recall also Proposition \ref{prop:DE}):

\begin{cor}[Basin of attraction]
\label{cor:basin}
Given the construction of Theorem \ref{thm:stabilizer-enc}, {the Pauli operators $R_{p,q}$ are pairwise commuting and identify a stabilizer subgroup, ${\cal G}$, on the physical gauge qubits. The latter identifies the basin of attraction ${\cal B}_{\bar{\rho}}$ of co-factor (gauge) states $\sigma$ for the physical encoder, that is, $\Phi_P(\rho_L\otimes \sigma)  = \bar{\rho} \in \cC$, for all $\sigma\in {\cal B}_{\bar{\rho}}$, where $\rho_L$ is the state on the first $n-r$ qubits to be encoded, and $\sigma$ has support on the subspace stabilized by ${\cal G}$.}

\end{cor}
\begin{proof} From the check matrix representation of the logical operators in Eqs.~\eqref{eq:Gz}-\eqref{eq:Gx}, the Pauli operators $R_{p,q}$ are defined via the rows of the two check matrices
\begin{align}
	R_z&=\begin{blockarray}{ccccc}
         {\scriptscriptstyle r-k} & {\scriptscriptstyle k} &  {\scriptscriptstyle r-k} & {\scriptscriptstyle k} \\
      \begin{block}{[cc|cc]l}
       0 & 0 & 0 & A_1^{\top} & \ {\scriptscriptstyle n-r} \\ 
      \end{block}
      & & &
    \end{blockarray} \in \Zbb_{2}^{(n-r)\times 2r} ,
    \label{eq:Rz} \\
    R_x&=\begin{blockarray}{ccccc}
        {\scriptscriptstyle r-k} & {\scriptscriptstyle k}& {\scriptscriptstyle r-k} & {\scriptscriptstyle k} \\
      \begin{block}{[cc|cc]l}
      D^{\top} & 0 & 0 & B^{\top} & \ {\scriptscriptstyle n-r} \\ 
      \end{block}
      & & &
    \end{blockarray} \in \Zbb_{2}^{(n-r)\times 2r} .
    \label{eq:Rx}
\end{align}
Since $R_z\Lambda R_z^{\top} = 0$, $R_x\Lambda R_x^{\top} = 0$, and $R_x\Lambda R_z^{\top} = 0$, it follows that the Pauli operators $R_{p,q}$ always commute pairwise and, therefore, form a stabilizer subgroup $\cal G$. Thus, the basin of attraction of the code, that is identified with the co-factor states $\sigma$ such that $\Tr(R_{p,q}\sigma )=1,$ $\forall p,q$, is determined by all the states with support on the subspace stabilized by $\cal G$. \qed
\end{proof}

A few additional observations can be made, regarding the structure of the $R_{p,q}$ operators and the basin of attraction:

{\bf (i)} In view of the check matrix representation of operators $R_{p,q}$ in Eqs. \eqref{eq:Rz}-\eqref{eq:Rx}, it follows that the basin of attraction \emph{always} contains the product state $|+\rangle^{\otimes r-k} \otimes |0\rangle^{\otimes k}$, where $|+\rangle$ denotes the $+1$-eigenstate of $X$. This can be considered as an easy to prepare, nominal gauge state for the initialization map $\Phi_L$.

{\bf (ii)} Since the operators $R_{p,q}$ in Eqs. \eqref{eq:Rz}-\eqref{eq:Rx} generate a stabilizer group, the basin of attraction has (at least) dimension $2^{r-\bar{r}}$, where $\bar{r}$ is the (row) rank of the matrix 
\[ \bar{R}\equiv \left[\begin{array}{c}
      	R_z\\
	\hline
	R_x
      \end{array}\right] \in \Zbb_{2}^{2(n-r)\times 2r}. \]
In the worst case where $\bar{R}$ is of full row rank, the basin of attraction has dimension $2^{3r-2n}$.

{\bf (iii)} If $D^{\top}$ in Eq. \eqref{eq:Rx} has $d_{x}$ all-zero columns, then the basis of attraction contains $d_{x}$ physical input qubits. Further, if the column block of $\bar{R}$,
\[ \left[\begin{array}{c}
      	A_{1}^{\top}\\
	B^{\top}
      \end{array}\right] \in \Zbb_{2}^{2(n-r)\times k}, \]
has $d_{z}$ all-zero columns, then the basis of attraction contains $d_{z}$ additional physical input qubits.

\section{Illustrative examples}

In this section, we exemplify the some of the general concepts and the construction of a FTDE presented in the previous section in a number of paradigmatic stabilizer quantum codes.  While we first consider some of the more standard stabilizer codes employed for encoding a single logical qubit, we discuss separately (in Sec.  in Sec. \ref{sub:tc}) the simplest example of a topological stabilizer code, Kitaev's toric code on the square lattice \cite{Kitaev2003,Dennis2002}. This is motivated by the additional special features that these codes enjoy -- namely, {\em locality} of all their stabilizer generators -- which also naturally suggest using an approach different than the general one based on their standard form for constructing the desired correction operators and dissipative maps.

\subsection{Standard stabilizer codes}
\label{sub:examples}

\subsubsection{Shor's 9-bit code} 

Shor's $9$-qubit code \cite{ShorCode} provided, historically, the first example of a {\em concatenated quantum code}, whereby the ability to correct arbitrary single-qubit errors is achieved by a ``nested'' repetition-code structure and linearity \cite{Nielsen-Chuang}. The code space is 
$$ \cH_\cC =\text{span}\Big\{ \frac{1}{2\sqrt{2}}  (|000\rangle \pm |111\rangle)^{ \otimes 3}   \Big \} ,$$
where the sign $+$ ($-$) corresponds to $|0_L\rangle$ ($|1_L\rangle$), respectively. That is, the two basis states have the form of a phase-flip code, where each of the three qubits of such a code has then been encoded in the bit-flip code described in Example \ref{ex:rep1}.  As mentioned in Remark \ref{rem1}, in the stabilizer language and seen as a subspace code, this corresponds to a $[[9,1,3]]$ code. Explicitly, in the standard form we also used in Theorem \ref{thm:stabilizer-enc}, the stabilizer generators are given by
\begin{align*}
S_{1}&=XXX II XXX I,\\
S_{2}&=XIIXXXXIX,\\
S_{3}&=IZIIIIIZI,\\
S_{4}&=IIZIIIIZI,\\
S_{5}&=IIIZIIIIZ,\\
S_{6}&=IIIIZIIIZ,\\
S_{7}&=ZIIIIZIII,\\
S_{8}&=ZIIIIIZII.
\end{align*}
The corresponding logical operators are
\begin{align*}
\overline{X}=XIIIIXXII,\qquad 
\overline{Z}=ZIIIIIIZZ.
\end{align*}
A set of correction operators $\{C_{i}\}_{i=1}^{8}$ such that $[C_{i},S_{j}]=2\delta_{ij}C_{i}S_{j}$ for all $i,j=1,2,\dots,8$ and $[C_{i},\overline{X}]=[C_{i},\overline{Z}]=0$ for all $i=1,2,\dots,8$, is then 
\begin{align*}
C_{1}&=IZIIIIIII,\\
C_{2}&=IIIZIIIII,\\
C_{3}&=IXIIIIIII,\\
C_{4}&=IIXIIIIII,\\ 
C_{5}&=IIIXIIIII,\\
C_{6}&=IIIIXIIII,\\
C_{7}&=IIIIIXIII,\\
C_{8}&=IIIIIIXII.
\end{align*}

Using Corollary \ref{cor:basin}, the basin of attraction of this code is found to consists of $8$-qubit co-factor states $\sigma \in {\cal B}_{\bar{\rho}}$ such that
\[  \Tr(R_{p,q}\sigma )= \Tr(I^{\otimes 4}\otimes X^{i}\otimes X^{i}\otimes Z^{j}\otimes Z^{j}\,\sigma )=1,\quad i,j=0,1. \]
This means that our proposed FTDE encodes an arbitrary state $\Phi_{8}\circ\cdots\circ\Phi_{1}(\rho_L \otimes \sigma)=\bar{\rho}\in \cC$, \emph{robustly} with respect to both initialization of $\sigma \in {\cal B}_{\bar{\rho}}$ and with respect to the application order of the maps.  In particular, any product state of the form $|\psi_{1}\rangle\otimes|\psi_{2}\rangle\otimes|\psi_{3}\rangle\otimes|\psi_{4}\rangle\otimes|+\rangle^{\otimes 2}\otimes|0\rangle^{\otimes 2}$, with $|\psi_{i}\rangle$, $i=1,2,3,4$, being arbitrary qubits, belongs to ${\cal B}_{\bar{\rho}}$.  As we noted in Sec. \ref{sub:sc}, Shor's code may be seen as a $[[9,1,3,3]]$ \cite{Poulin2005} or even a $[[9,1,4,3]]$ subsystem code \cite{Breuckmann-thesis}. Thus, since an identification between the gauge degrees of freedom and the physical qubits exists before encoding,   robustness against initialization errors in the state of these gauge qubits is also available through unitary encoding procedures although, in unitary settings, it is usually necessary to trade this extra tolerance for a reduced complexity of the encoding circuit itself \cite{Klappenecker2009}.

\subsubsection{Steane's 7-bit code}

Steane's code, independently discovered as one of the first known quantum codes \cite{SteaneCode}, also achieves QEC against arbitrary single-qubit errors, however without resorting to a concatenated structure. The stabilizer generators for this $[[7,1,3]]$ code, when the check matrix is put in standard form as before, read
\begin{align*}
S_{1}&=XIXXXII,\\
S_{2}&=XXIXIXI,\\
S_{3}&=IXXXIIX,\\
S_{4}&=ZZIIZIZ,\\
S_{5}&=ZIZIIZZ,\\
S_{6}&=IIIZZZZ, 
\end{align*}
whereas he corresponding logical operators are
\begin{align*}
\overline{X}=XXXIIII,\qquad 
\overline{Z}=ZIIIZZI.
\end{align*}
A set of correction operators $\{C_{i}\}_{i=1}^{6}$ such that $[C_{i},S_{j}]=2\delta_{ij}C_{i}S_{j}$ for all $i,j=1,2,\dots,6$ and $[C_{i},\overline{X}]=[C_{i},\overline{Z}]=0$ for all $i=1,2,\dots,6$, is 
\begin{align*}
C_{1}&=ZIZZIII,\\
C_{2}&=ZZIZIII,\\
C_{3}&=IZZZIII,\\
C_{4}&=IXIIIII,\\ 
C_{5}&=IIXIIII,\\
C_{6}&=IIIXIII.
\end{align*}

The basin of attraction now consists of $6$-qubit co-factor (gauge) states $\sigma$ such that
\[  \Tr(R_{p,q}\sigma )= \Tr(X^{i}\otimes X^{i}\otimes I \otimes Z^{j}\otimes Z^{j} \otimes I\,\sigma )=1,\quad i,j=0,1. \]
In particular, any product state of the form $|+\rangle^{\otimes 2}\otimes |\psi_{1}\rangle\otimes|0\rangle^{\otimes 2}\otimes |\psi_{2}\rangle$, with $|\psi_{i}\rangle$, $i=1,2$, being arbitrary qubits, belongs to the basin of attraction. Unlike for Shor's code, no gauge symmetry (hence no gauge qubits) are present in this case. Therefore, dissipative encoders provide in this case the only means to achieve error-tolerant initialization.

\subsubsection{The 5-qubit perfect code}

The $5$-bit code provides the simplest example of a stabilizer code that is not {\em additive}, namely, one that is not constructed from classical codes with some special properties \cite{Nielsen-Chuang}. It is also provably the smallest quantum code that can correct arbitrary single-qubit errors, hence often referred to as a ``perfect'' code \cite{PerfectCode}. The stabilizer generators for this $[[5,1,3]]$ code are, in their standard form, 
\begin{align*}
S_{1}&=YYZIZ,\\
S_{2}&=XIXZZ,\\
S_{3}&=XZZXI,\\
S_{4}&=YZIZY, 
\end{align*}
with the corresponding logical operators 
\begin{align*}
\overline{X}=X\,Z\,I\,I\,Z,\qquad 
\overline{Z}=Z\,Z\,Z\,Z\,Z.
\end{align*}
A set of correction operators $\{C_{i}\}_{i=1}^{4}$ such that $[C_{i},S_{j}]=2\delta_{ij}C_{i}S_{j}$ for all $i,j=1,2,3,4$ and $[C_{i},\overline{X}]=[C_{i},\overline{Z}]=0$ for all $i=1,2,3,4$, is 
\begin{align*}
C_{1}&=ZIXIX,\\
C_{2}&=YIXXX,\\
C_{3}&=YXXXI,\\
C_{4}&=ZXIXI.
\end{align*}

In this case, the basin of attraction consists of $4$-qubit co-factor (gauge) states $\sigma$ such that
\[  \Tr(R_{p,q}\sigma )= \Tr((Z^{i}Z^{j})\otimes Z^{j}\otimes Z^{j}\otimes (Z^{i}Z^{j})\,\sigma )=1,\quad i,j=0,1. \]
In particular, the product state $|0\rangle^{\otimes 4}$ belongs to the basin of attraction. Again, no representation with gauge qubits is available for this code and hence a dissipative quantum circuit $\Phi_{4}\circ\cdots\circ\Phi_{1}$ is the only way to warrant a non-trivial robustness in encoding.

\subsection{Topological stabilizer codes: Kitaev's toric code}
\label{sub:tc}

The toric code, $\cC_T$, in two dimensions employs a 2D lattice of $2L^2$ physical qubits arranged on the edges of the squares of a grid \footnote{The name ``toric code'' is due to identifying the eastern and western border qubits as neighboring one another, while similarly for the northern and southern border qubits; hence, the geometry and topology induced by the adjacency of qubits is that of a flat 2-torus.}. As we already remarked, a key feature of this topological code is that \emph{all} the stabilizer operators are geometrically local. Neighboring quartets of qubits either surround a square of the grid (``plaquette'') or surround the intersection of a vertical and a horizontal line of the grid (``vertex''). To each plaquette $p$ and to each vertex $v$ (see Fig. \ref{fig:labels}a), we assign a four-body stabilizer acting on the corresponding qubits defined, respectively, by
\begin{equation}
H_p \equiv (Z^{\otimes 4})_p\otimes \identity_{\bar{p}}, \qquad 
H_v \equiv (X^{\otimes 4})_v\otimes \identity_{\bar{v}}, 
\end{equation}
where $\bar{p}$ denotes the complement to $p$, and similarly for $\bar{v}$.  Since each plaquette overlaps with an even number of systems in any vertex, we have $[H_p,H_v]=0$ for all $p$ and $v$. One consequence of the toroidal geometry is that the above plaquette and vertex Hamiltonians are not algebraically independent, since 
\begin{equation}
\label{eq:redundant}
 \prod_p H_p =\identity\,\,\,\textup{and}\,\,\, \prod_v H_v =\identity.
\end{equation}
These algebraic dependencies lead to the definition of a set of logical operators corresponding to a two-qubit logical subspace. This space constitutes the toric code and may be defined as the space of vectors $\ket{\psi}$ satisfying $H_p\ket{\psi}=H_v\ket{\psi}=\ket{\psi}$ for all $p$ and $v$. Physically, $\cC_T$ corresponds to the (four-fold degenerate) ground space of the \emph{toric-code Hamiltonian} $H_{T}\equiv -(\sum_p H_p +\sum_v H_v).$ Later on, we will use the fact that it suffices to associate $\cC_T$ to the $+1$-eigenspace of all but one $H_p$ and all but one $H_v$. This follows from the algebraic redundancy expressed in Eq. (\ref{eq:redundant}).

\begin{figure}[t]
    \centering
    \subfloat[Plaquette and vertex stabilizer operators]
    {{\includegraphics[width=0.3\columnwidth,viewport= 0 0 760 760,clip]{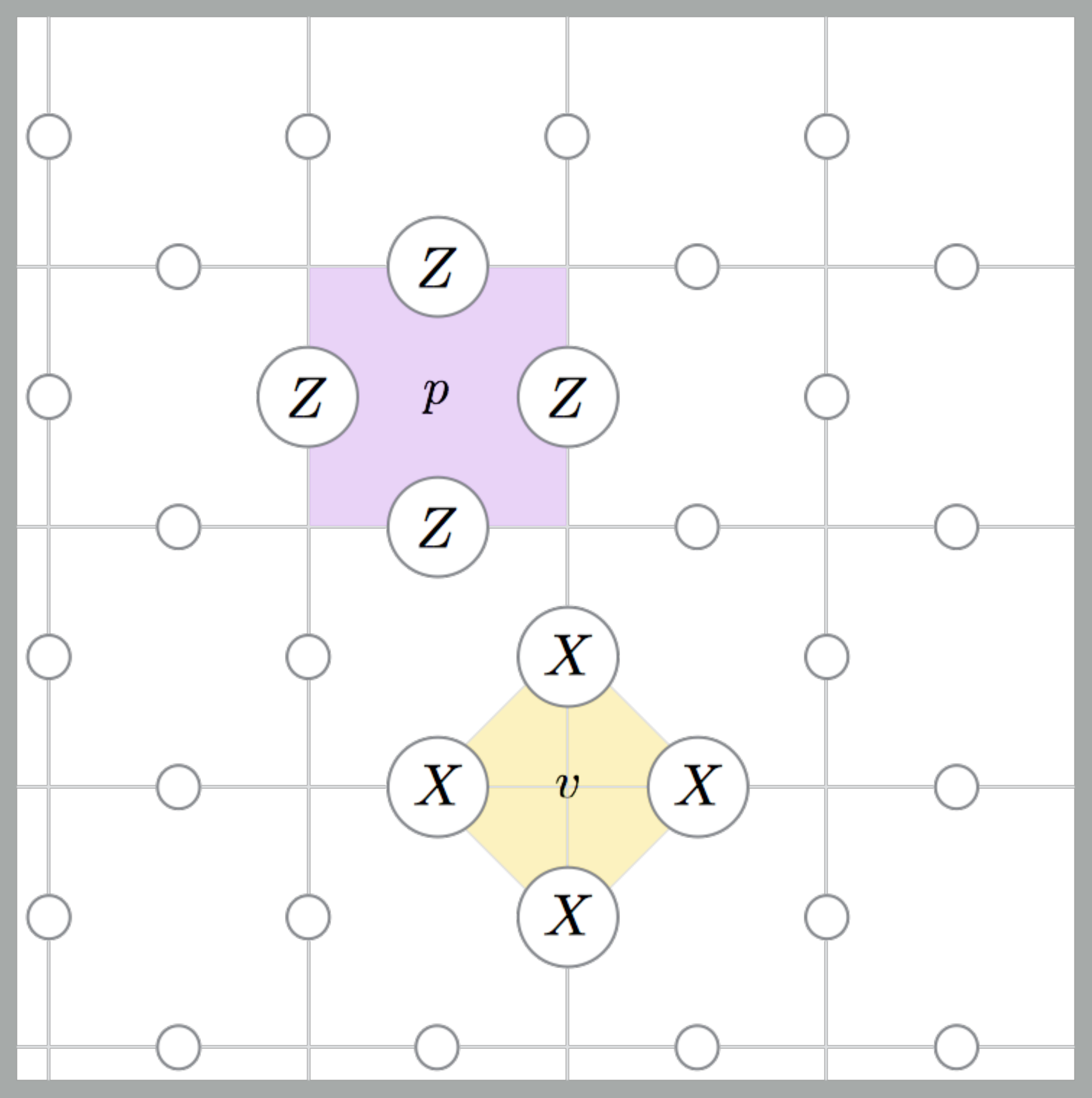}}}%
    \qquad\qquad
    \subfloat[Qubit labeling scheme]
    {{\includegraphics[width=0.3\columnwidth,viewport= 0 0 760 760,clip]{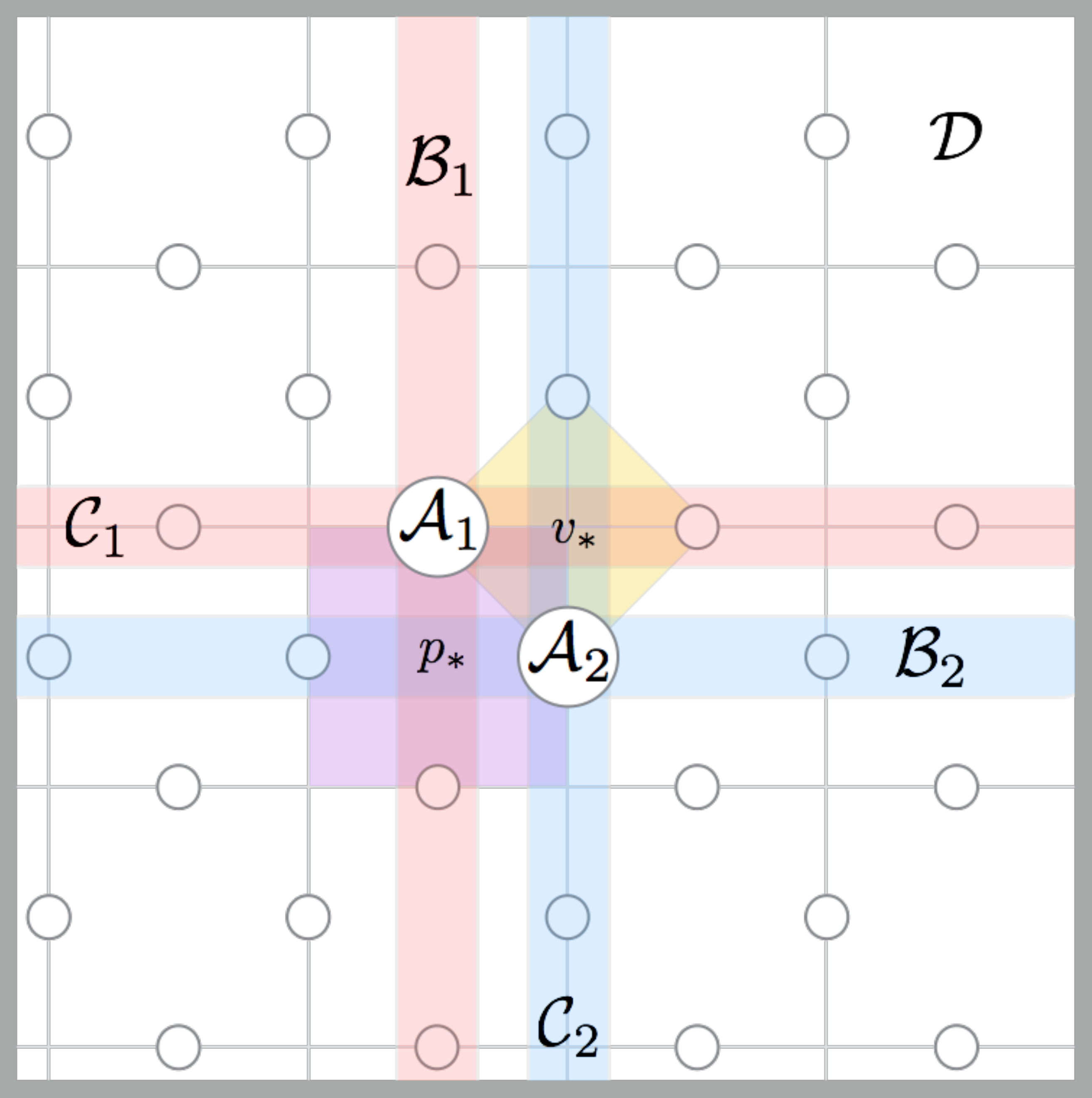}}}%
    \caption[Stabilizers and labeling scheme for 2D toric code]{(a) The four-body stabilizer operators of the toric code act on vertices $v$ as $X^{\otimes 4}$ or on plaquettes as $Z^{\otimes 4}$. (b) $\mathcal{A}_1$ and $\mathcal{A}_2$ are the two upload qubits whose initial state is mapped into the code. $\mathcal{B}_1$ and $\mathcal{C}_1$ are  the systems (in addition to $\mathcal{A}_1$) on which the code's logical operators for the first encoded qubit are defined to act. The same holds for $\mathcal{B}_2$ and $\mathcal{C}_2$ with respect to the second encoded qubit. $\mathcal{D}$ denotes the remaining qubits. The systems $\mathcal{B}_1$, $\mathcal{C}_1$, $\mathcal{B}_2$, and $\mathcal{C}_2$ must be properly initialized in order to achieve a faithful encoding.}%
    \label{fig:labels}%
\end{figure}

\begin{figure}[h]
    \centering
    \subfloat[Logical $X$s of toric code]
    {{\includegraphics[width=0.3\columnwidth,viewport= 0 0 760 760,clip]{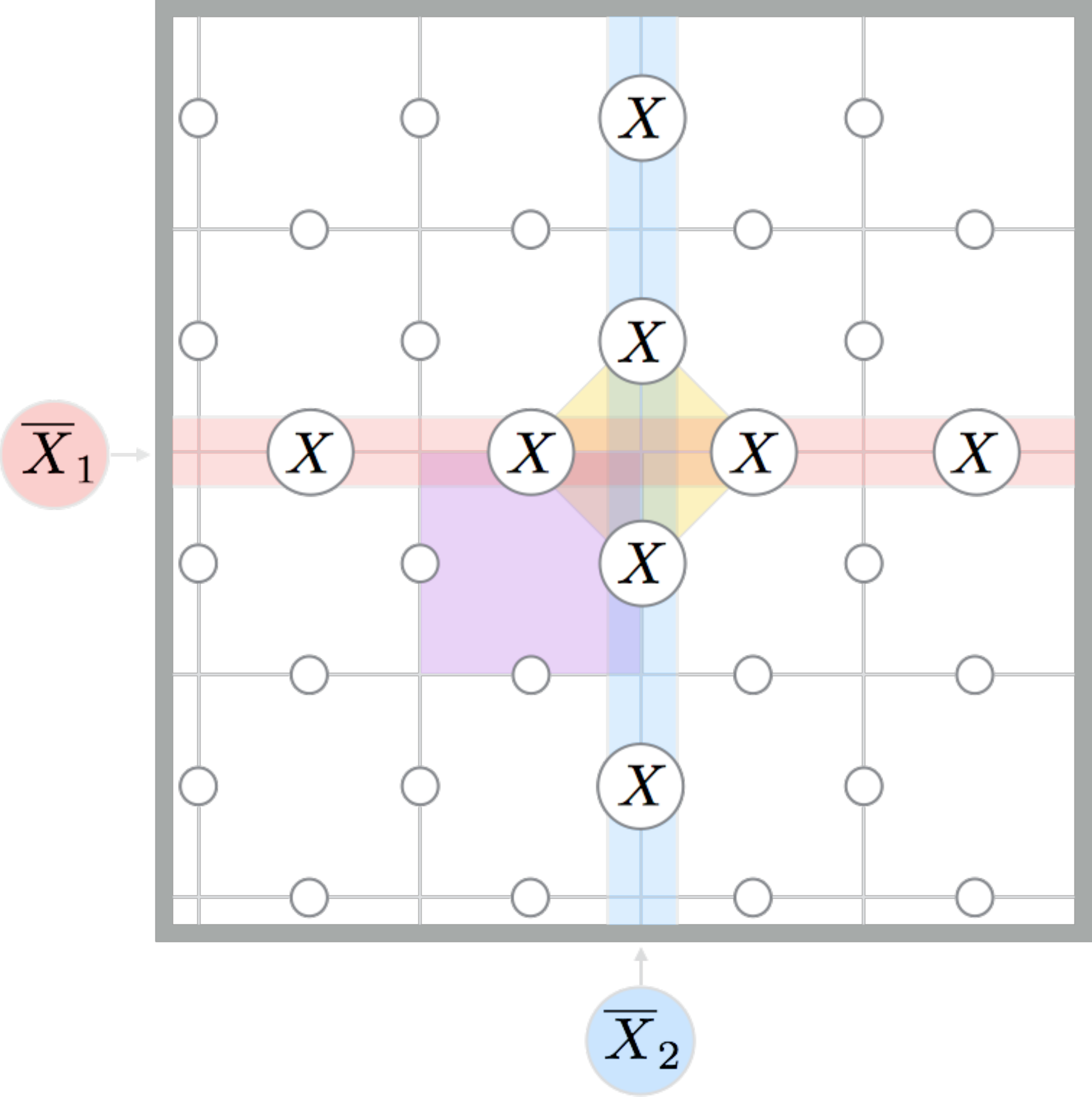} }}%
    \qquad\qquad
    \subfloat[Logical $Z$s of toric code]
    {{\includegraphics[width=0.3\columnwidth,viewport= 0 -109 760 760,clip]{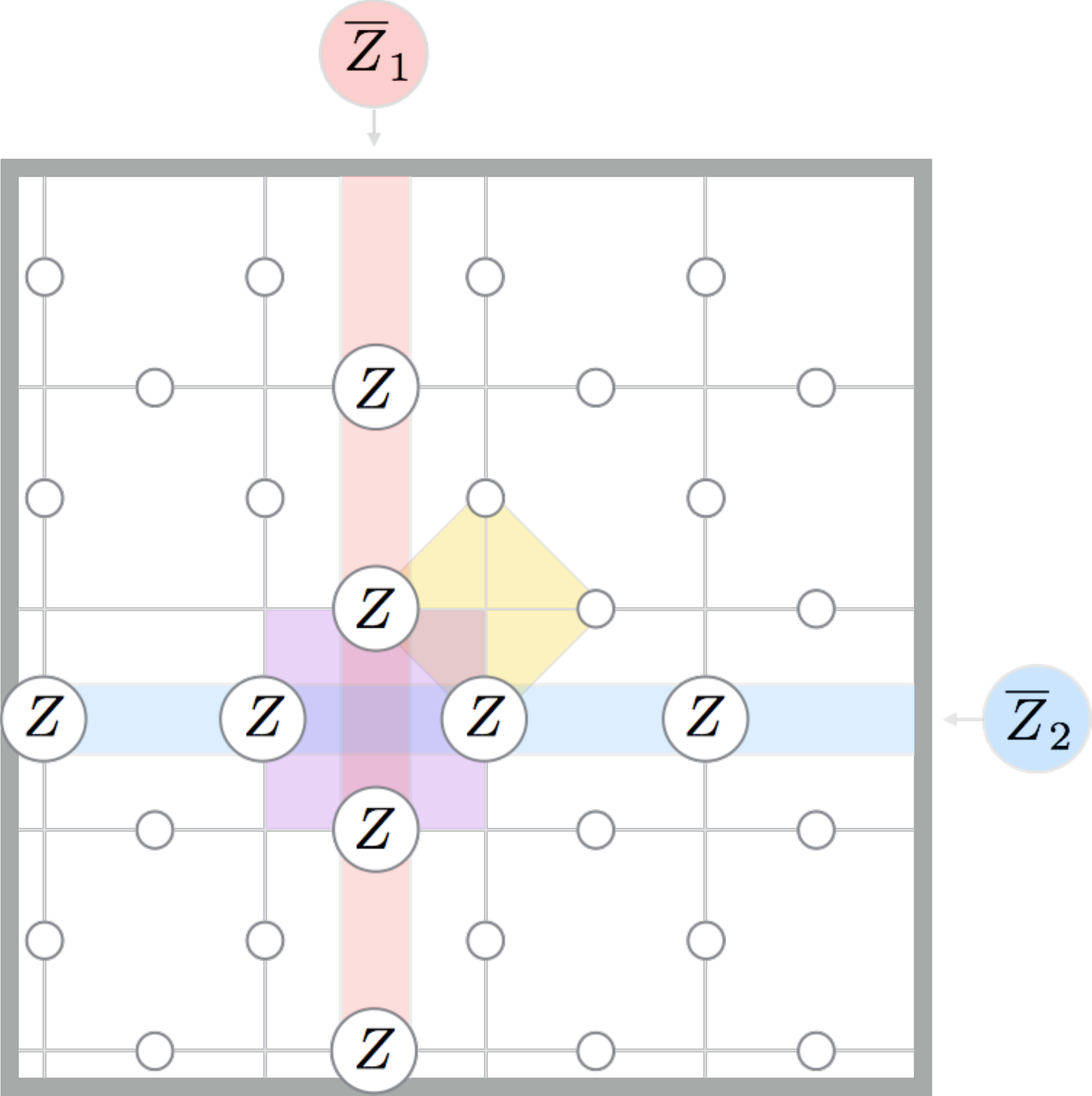}}}%
    \caption[Logical operators for toric code]{A definition of logical operators for the toric code. 
    Logical operators consist of topologically non-trivial loops of local bit-flip or local phase-flip errors.}%
    \label{fig:logical}
\end{figure}

In \cite{Dengis2014}, the authors provide a Liouvillian, constructed as a sum of plaquette- and vertex-acting terms, which generates a CDE from two localized physical qubits into $\cC_T$. Following their labeling scheme (see Fig. \ref{fig:labels}b), the two neighboring upload qubits $\mathcal{A}_1$ and $\mathcal{A}_2$ are chosen so as to share a plaquette and vertex, which are labeled $p_*$ and $v_*$. Then, the qubits of the vertical and horizontal strips which pass through $p_*$ (except for $\mathcal{A}_1$ and $\mathcal{A}_2$) are each prepared in $\ket{+}$. These strips are labeled $\mathcal{B}_1$ and $\mathcal{B}_2$, respectively. Similarly, the qubits of the bands passing through $v_*$ are each prepared in $\ket{0}$. If these strips are labeled $\mathcal{C}_1$ and $\mathcal{C}_2$, respectively, {this results in the state $|\phi\rangle_{\mathcal{B}\mathcal{C}} \equiv |+\rangle^{\otimes L-1}_{\mathcal{B}_1}\otimes |+\rangle^{\otimes L-1}_{\mathcal{B}_2}\otimes |0\rangle^{\otimes L-1}_{\mathcal{C}_1}\otimes |0\rangle^{\otimes L-1}_{\mathcal{C}_2}$. The logical operators are chosen as (see Fig. \ref{fig:logical}): 
\begin{align*}
\overline{X}_1&\equiv X_{\mathcal{A}_1}\otimes X^{\otimes L-1}_{\mathcal{B}_1}\otimes\identity_\cD, \qquad 
\overline{Z}_1\equiv Z_{\mathcal{A}_1}\otimes Z^{\otimes L-1}_{\mathcal{C}_1}\otimes\identity_\cD, \\  
\overline{X}_2&\equiv X_{\mathcal{A}_2}\otimes X^{\otimes L-1}_{\mathcal{B}_2}\otimes\identity_\cD, \qquad 
\overline{Z}_2\equiv X_{\mathcal{A}_2}\otimes Z^{\otimes L-1}_{\mathcal{C}_2}\otimes\identity_\cD . 
\end{align*}
The authors of \cite{Dengis2014} show that, with this initialization, the state of $\mathcal{A}_1\otimes\mathcal{A}_2$ is driven towards the corresponding state of $\cC_T$} by means of their constructed Lindblad dynamics. After reviewing the CPTP correction maps used to construct such a Lindblad dynamics, we show that a judicious ordering of these maps does, in fact, constitute a FTDE. 

The stabilizer generators are $\{H_p,H_v\}$, where the set ranges over all plaquettes and vertices, except for $p_*$ and $v_*$ as specified above. Correction maps $\Phi_p$ and $\Phi_v$ are associated to each of these stabilizer generators, { for a total number of $t\equiv 2(L^2-1)$ maps. 
Let plaquettes and vertices be labeled according to their lattice coordinates with respect to $p_*$ and $v_*$, so that the plaquette $p_{\alpha,\beta}$ and vertex $v_{\alpha,\beta}$ lie $\alpha$ sites north and $\beta$ sites east of $p_*$ and $v_*$, respectively. If we use a tensor product structure where, for each plaquette and vertex system, the north, east, south, and west qubits are $N\otimes E\otimes S\otimes W$, we further define the following (unitary) correction operators:}
\begin{align*}
C_{p_{\alpha,0}}\equiv (\identity\otimes X \otimes \identity\otimes \identity)_p\otimes \identity_{\bar{p}},\nonumber\\ 
C_{p_{\alpha,\beta}}\equiv (X\otimes \identity \otimes \identity\otimes \identity)_p\otimes \identity_{\bar{p}},\nonumber\\
C_{v_{0,\beta}}\equiv (\identity\otimes \identity \otimes Z\otimes \identity)_v\otimes \identity_{\bar{v}},\nonumber\\
C_{v_{\alpha,\beta}}\equiv (\identity\otimes \identity\otimes \identity \otimes Z)_v\otimes \identity_{\bar{v}}.
\end{align*}
Finally, let the Kraus operators for the plaquette and vertex correction maps $\Phi_{p_{\alpha,\beta}}$ and $\Phi_{v_{\alpha,\beta}}$ be labeled $K_{p_{\alpha,\beta}}^{(i)}$ and $K_{v_{\alpha,\beta}}^{(i)}$, respectively. The scheme devised in \cite{Dengis2014} {for continuously correcting residual errors in $\mathcal{A}_1$ and $\mathcal{A}_2$ suggests a choice of ordering for these dissipative maps} as depicted in Fig. \ref{fig:correction}. From the construction in the proof of Theorem \ref{thm:stabilizer-enc}, we expect that the key feature of this choice of correction operators and ordering is that subsequent correction operators commute with all previous stabilizer operators. 

\begin{figure}[t]
    \centering
    \subfloat[Vertex correction maps]
    {{\includegraphics[width=0.3\columnwidth,viewport= 40 40 735 735,clip]{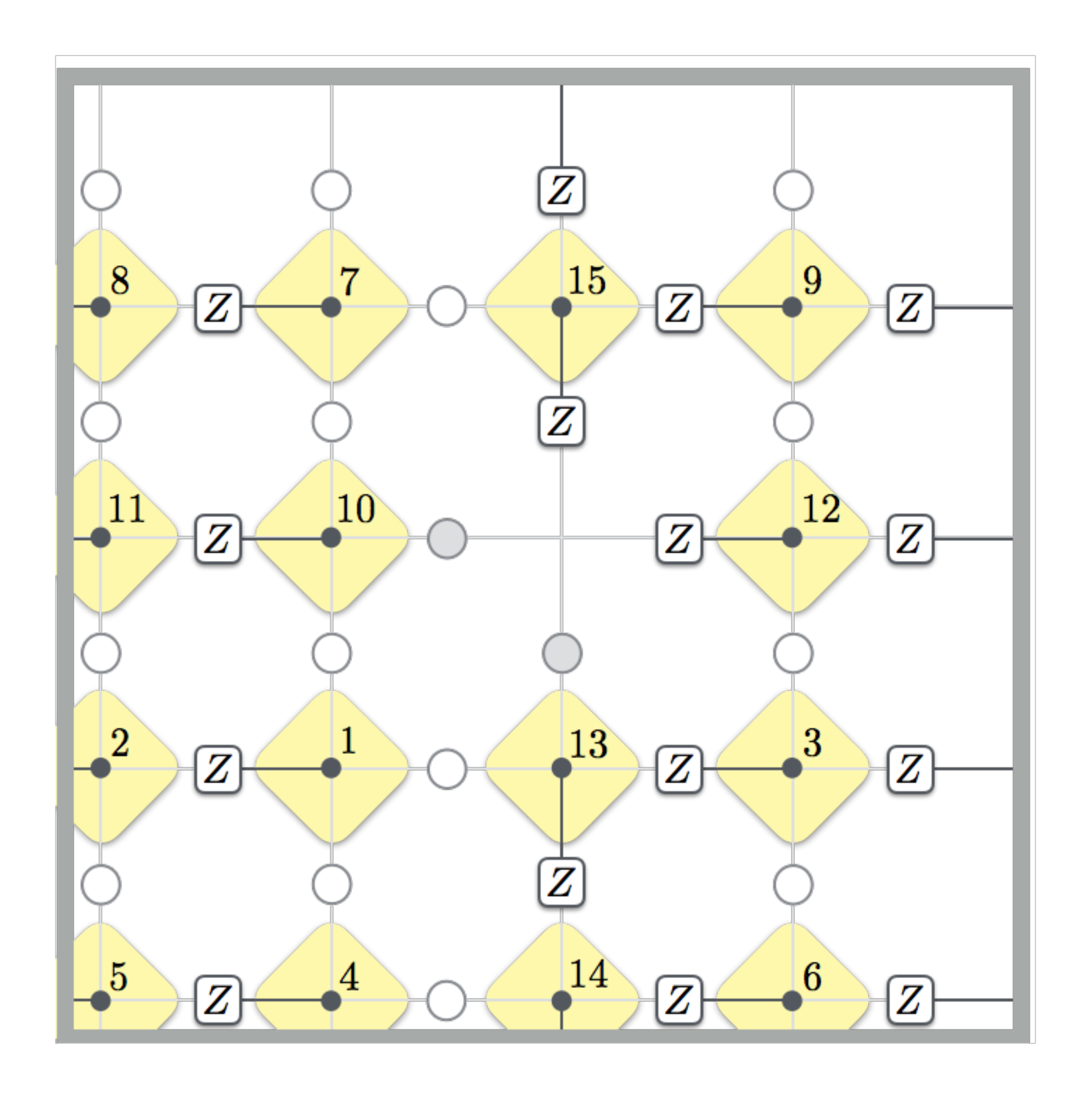} }}
    \qquad\qquad
    \subfloat[Plaquette correction maps]
    {{\includegraphics[width=0.3\columnwidth,viewport= 0 0 760 760,clip]{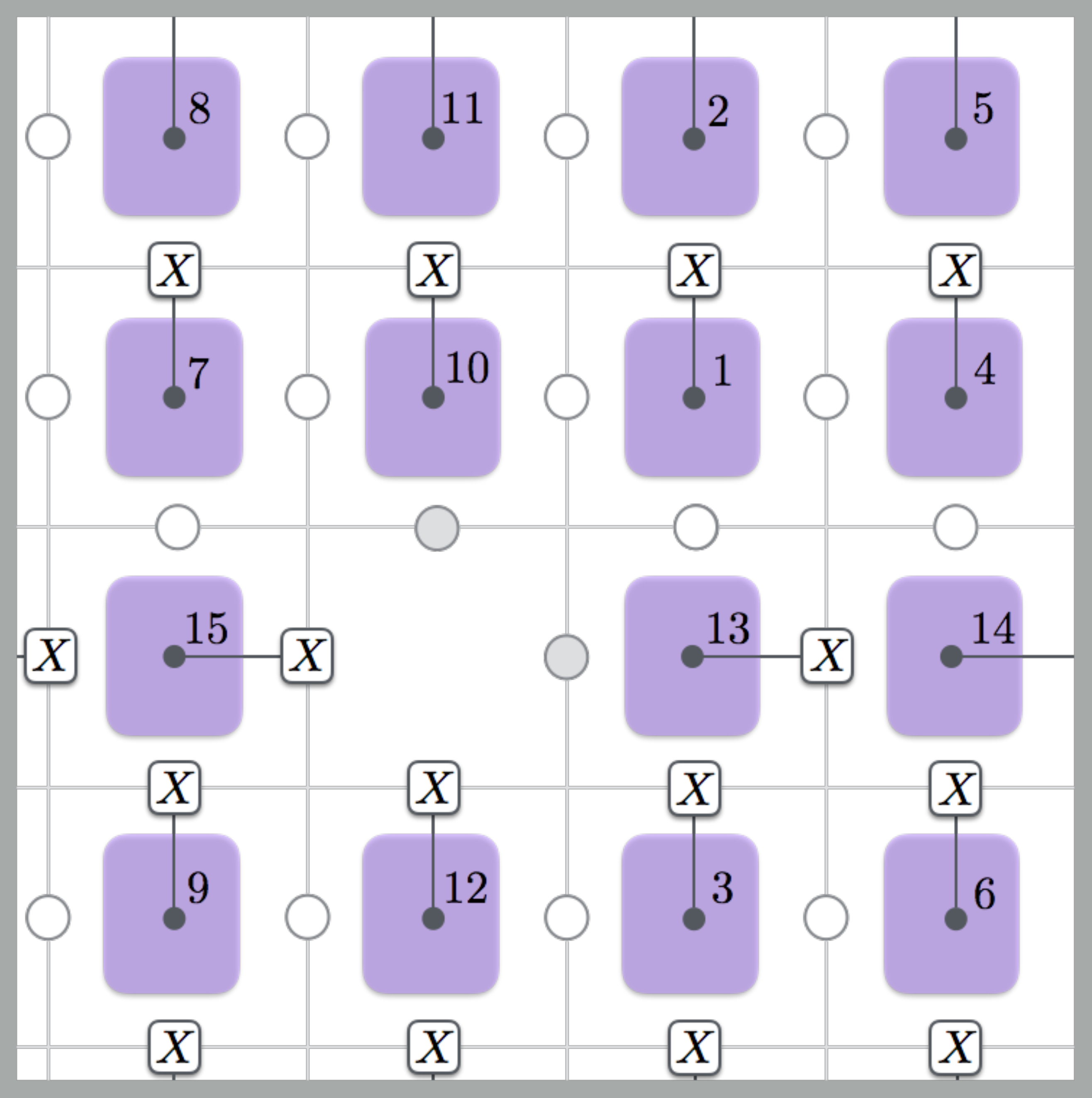}}}%
    \caption[Sequence of correction maps for finite-time encoding]{Sequences of CPTP maps for FTDE. The ordering of the correction maps and the location of correction operators are chosen so that 1) correction operators commute with all four logical operators and 2) subsequent correction operators are applied where no correction map has acted previously.}%
    \label{fig:correction}%
\end{figure}

Each plaquette correction map commutes with each vertex correction map, since exchange of their Kraus operators can, at most, accrue an irrelevant global phase (which cancels in the superoperator). Without loss of generality, we consider the vertex maps to act first. As seen in Fig. \ref{fig:correction}, the ordering among the vertex (respectively, plaquette) correction maps is chosen such that each subsequent correction operator acts where no previous correction map (and hence stabilizer operator) has acted. This verifies that subsequent correction operators commute with all previous stabilizer operators, as desired. 

{Let ${\bm i} \equiv (i_1,\ldots, i_t) \in \{ 0,1 \}^t$ specify the $i=0,1$ values of the plaquette and vertex Kraus operators. Then, moving all correction operators to the right of the stabilizers, each Kraus operator of the FTDE $\Phi_P$ may be written as follows
\begin{align*}
K^{({\bm i})}=(\Pi_1C_1^{i_1})\ldots(\Pi_{t}C_t^{i_t})
=\Pi^{+}(C^{i_1}_{1}\ldots C_{t}^{i_t})
=\Pi^{+}C^{({\bm i})},
\end{align*}
where $\Pi_j\equiv \frac{1}{2}(\identity + H_j)$ and $\Pi^{+}$ is the projector onto the support of $\cC_T$.

The encoder can then be written as 
$$\Phi_{P}(\cdot)= \Phi_1\circ \ldots \circ \Phi_t (\cdot)= 
\Pi^{+} \Big(\sum_{{\bm i}} C^{({\bm i})}\cdot C^{({\bm i}) \dagger} \Big) \Pi^{+}.$$ 
This verifies that the range of the encoder is contained in the code itself.  While the existence of a FTDE for $\cC_T$ is implied by 
Theorem \ref{thm:stabilizer-enc}, the above construction further explicitly shows how this encoding can be achieved using dissipative dynamics that is QL relative to the natural (four-local) neighborhood structure. We note that, while the number of required CPTP maps scales as $L^2$, if we apply commuting maps in parallel, it is possible to equivalently implement $\Phi_P$ 
in a time that scales as $L$. This matches the scaling of the CDE of \cite{Dengis2014} (see also \cite{Pastawski2014}), with the important advantage that in our case we can guarantee \emph{exact} preparation in a finite evolution time, as long as all $\Phi_j$ are implemented correctly.}

Finally, we determine the initializations that ensure a two-qubit state $\rho_L$ on $\mathcal{A}_1\otimes\mathcal{A}_2$ to be encoded into $\overline{\rho}\in \cC_T$. Let the initial state be $\rho_L\otimes \sigma = \rho_{\mathcal{A}_1\mathcal{A}_2}\otimes\sigma_{\mathcal{B}\mathcal{C}\mathcal{D}}$. The basin of attraction for $\sigma$ is determined by the analogue of Eq. (\ref{eq:encodingrelation}), namely, 
\begin{equation*}
\tr{} ({X_{\mathcal{A}_1}^iZ_{\mathcal{A}_1}^jX_{\mathcal{A}_2}^k Z_{\mathcal{A}_2}^l\rho_L})=\tr{}\Big({\overline{X}_1^i\overline{Z}_1^j\overline{X}_2^k\overline{Z}_2^l\Phi_P(\rho_{\mathcal{A}_1\mathcal{A}_2}\otimes\sigma_{\mathcal{B}\mathcal{C}\mathcal{D}})} \Big), 
\quad \forall i,j,k,l \in \{0,1\}. 
\end{equation*}
Since the correction operators and the stabilizers commute with the logical operators, we have $\Phi_P^{\dagger}(\overline{X}_1^i\overline{Z}_1^j\overline{X}_2^k\overline{Z}_2^l)=\overline{X}_1^i\overline{Z}_1^j\overline{X}_2^k\overline{Z}_2^l$. With this, the above equation simplifies to
\begin{align*}
\tr{} ({X_{\mathcal{A}_1}^iZ_{\mathcal{A}_1}^jX_{\mathcal{A}_2}^k Z_{\mathcal{A}_2}^l\rho_L}) &=
{
\tr{} \Big( {(X_{\mathcal{A}_1}^iZ_{\mathcal{A}_1}^jX_{\mathcal{A}_2}^k Z_{\mathcal{A}_2}^l\rho_L)\otimes ((X^i Z^j X^k Z^l)_{\mathcal{B}\mathcal{C}}^{\otimes L-1}\sigma_{\mathcal{B}\mathcal{C}\mathcal{D}})} \Big) }\\
&= {
\tr{} \Big( {X_{\mathcal{A}_1}^iZ_{\mathcal{A}_1}^jX_{\mathcal{A}_2}^k Z_{\mathcal{A}_2}^l\rho_L\,} \Big)\tr{}{\Big( (X^i Z^j X^k Z^l)_{\mathcal{B}\mathcal{C}}^{\otimes L-1}\sigma_{\mathcal{B}\mathcal{C}}} \Big),}
\end{align*}
where in the last step, we have traced out $\mathcal{D}$ to obtain $\sigma_{\mathcal{B}\mathcal{C}}=\tr{}({\mathcal{D}}{\sigma_{\mathcal{B}\mathcal{C}\mathcal{D}}})$. Hence, the state of system $\mathcal{D}$ does not affect the encoding. This determines the basin of attraction to be states $\sigma_{\mathcal{B}\mathcal{C}\mathcal{D}}$ with reduced state on $\mathcal{B}\mathcal{C}$ satisfying $\tr{} ({(X^i Z^j X^k Z^l)_{\mathcal{B}\mathcal{C}}^{\otimes L-1}\sigma_{\mathcal{B}\mathcal{C}}})=1$, for all $i,j,k,l$. A sufficient choice of initial state, as used in  \cite{Dengis2014}, is the pure state $\ket{\phi}_{\mathcal{B}\mathcal{C}}=\ket{+}^{\otimes L-1}_{\mathcal{B}_1}\ket{0}^{\otimes L-1}_{\mathcal{C}_1}\ket{+}^{\otimes L-1}_{\mathcal{B}_2}\ket{0}^{\otimes L-1}_{\mathcal{C}_2}$ introduced before. This demonstrates that a suitable initialization can be implemented locally. The full basin of attraction can be described as comprising any density operator with support in the $+1$-eigenspace of each of the four commuting operators $X_{\mathcal{B}_1}^{\otimes L-1}$, $Z_{\mathcal{C}_1}^{\otimes L-1}$, $X_{\mathcal{B}_2}^{\otimes L-1}$, and $Z_{\mathcal{C}_2}^{\otimes L-1}$.

\section{Conclusion}
\label{sec:end}

The ability to effectively map -- \emph{encode} -- abstractly defined quantum information into logically represented degrees of freedom associated to a quantum code is an essential prerequisite for achieving noise-protected quantum information processing in physical devices. While standard encoding procedures employ discrete-time dynamics as realized by finite sequences of unitary quantum gates, we focused on formalizing the notion of a dissipative (Markovian) quantum encoder. In particular, we argued how the use of discrete-time dissipative dynamics -- a dissipative quantum circuit in the sense of \cite{johnson-FTS} -- can combine advantageous features from both the unitary setting and the continuous-time setting considered in \cite{Dengis2014,encoding-CDC}. Specifically, such encoders can afford extra robustness against a class of initialization errors thanks to the non-trivial basins of attraction that their dissipative nature affords, while retaining the possibility of exact convergence in a finite number of steps that unitary circuits also enjoy. As a main result, we have showed that finite-time dissipative encoders may always be constructed in principle for the important class of stabilizer quantum error-correcting codes. For stabilizer codes which, like Steane's 7-bit code or the 5-bit code, fail to possess gauge degrees of freedom in their subsystem (operator error correction) form, dissipative encoders provide the only means for achieving nontrivial robustness against initialization errors, with the details depending on the structure of the encoder's basin of attraction.  

A number of questions are left to future investigations. From both a conceptual and a practical standpoint, a particularly relevant question is to determine the extent to which finite-time convergence may be affected by relevant kinds of errors in implementing the required encoding maps. While stability results under local perturbations have been established, in the quantum setting, for continuous-time Markovian dynamics exhibiting rapid mixing \cite{Lucia}, no results are available for finite-time-stable discrete-time dynamics to the best of our knowledge.  Also of interest would be the design of explicit dissipative quantum encoders tailored to specific implementation platforms -- notably, circuit QED architectures, for which explicit protocols for implementing quantum maps using an ancilla qubit with QND readout and adaptive control were recently proposed \cite{Liang}.

\acknowledgments
Work at Dartmouth was partially supported by the US NSF under grant No. PHY-1620541 and the US DOE, Office of Science, Office of Advanced Scientific Computing Research, Accelerated Research for Quantum Computing program. G.B. and F.T. have been partially funded by the QUINTET, QFUTURE and QCOS projects of the Universit\`a degli Studi di Padova.
	

\begin{appendix}
\section{Symplectic representation for stabilizers codes}
\label{sympl}

We recall that one can conveniently represents the generators $\{S_{k}\}_{k=1}^{r}$ of a stabilizer group using the so-called symplectic representation and the {\em check matrix} \cite[Sec.~10.5]{Nielsen-Chuang}. For qubits, the latter is an $r\times 2n$ matrix whose rows correspond to the generators $S_1$ through $S_r$, and the columns to the set of generators of the Pauli group $\cP_n$ given by single-qubit $X$ and $Z$ operators. The left-hand side of the matrix contains 1s to indicate which generators contain $X$s, and the right-hand side contains 1s to indicate which generators contain $Z$s; the presence of a 1 on both sides indicates a $Y$ in the generator. 

One of the main interests in this representation is that it allows to check commutativity as a linear algebraic property, through a symplectic product. Let us define the $2n\times2n$ matrix
\[ \Lambda \equiv \begin{bmatrix} 0 & I \\ I & 0 \end{bmatrix}, \]
where the $I$ matrices on the off-diagonals are $n\times n$. It can be shown that two elements of the Pauli group, say $g_1$ and $g_2$, commute, if  and only if the corresponding row vectors  in the check matrix representation, say $G_1$ and $G_2$, satisfy $G_1\Lambda G_2^{\top}=0$ (where arithmetic is done modulo two).

\section{Technical proofs}

\subsection{Co-tensor factor states must be independent of logical information}
\label{app:cofactor}

\begin{lemma}
\label{lem:basis}
Consider two decompositions $\cH_{L'}\otimes \cH_{F'}\oplus \cH_{R'}$ and $\cH_{L''}\otimes \cH_{F''}\oplus \cH_{R''}$ with $\dim(\cH_{L'})=\dim(\cH_{L''})$, and full rank states $\tau_{F'}\in\cD(\cH_{F'})$, $\tau_{F''}\in\cD(\cH_{F''})$. 
Suppose that for any pair of \emph{pure} initializations in $\cH_{L'}$ with co-factor state $\tau_{F'}$, and in $\cH_{L''}$ with co-factor state $\tau_{F''}$, it holds
\begin{align}
\label{eq:compsubpure}
	\begin{cases}
	\Pi_{L'F'}(\rho_{L}\otimes\tau_{F''}\oplus 0_{R''})\Pi_{L'F'}\simeq \rho_{L}\otimes\tilde\tau_{F'}(\rho_{L})\oplus 0_{R'},\\
	\Pi_{L''F''}(\rho_{L}\otimes\tau_{F'}\oplus 0_{R'})\Pi_{L''F''}\simeq \rho_{L}\otimes\tilde\tau_{F''}(\rho_{L})\oplus 0_{R''}
	\end{cases}\ \ \forall \rho_{L}\in\cD(\cH_{L}),
\end{align}
with $\Pi_{L'F'}$, $\Pi_{L''F''}$ as in Definition \ref{def:comp} and $\tilde{\tau}_{F'}(\rho_{L}),\, \tilde{\tau}_{F''}(\rho_{L})\geq 0$ that are, in general, functions of~$\rho_{L}$. Then the two initializations are compatible -- i.e., $\tilde{\tau}_{F'}(\rho_{L}),\, \tilde{\tau}_{F''}(\rho_{L})$ are actually independent of~$\rho_L$.
\end{lemma}
\begin{proof}
We first show that if the first condition in Eq. \eqref{eq:compsubpure} holds true, then the same condition holds true for every state $\rho_{L}\in\cD(\cH_{L})$ with $\tilde{\tau}_{F'}(\rho_{L})$ independent of $\rho_{L}$. To this aim, let $\{\ket{\psi_{i}}\}_{i}$ be an orthonormal basis in $\cH_{L}$ and define $\rho''_{i}\simeq\ket{\psi_{i}}\bra{\psi_{i}}\otimes\tau_{F''}\oplus 0_{R''}$, $\tilde{\rho}_{i}'\simeq\ket{\psi_{i}}\bra{\psi_{i}}\otimes\tilde\tau_{F'i}\oplus 0_{R'}$. By assumption, we have
\begin{align*}
	\Pi_{L'F'} (\alpha \rho''_{i}+\beta \rho''_{j})\Pi_{L'F'} &= \alpha \tilde{\rho}_{i}' + \beta \tilde{\rho}_{j}'\\
	& \simeq \alpha (\ket{\psi_{i}}\bra{\psi_{i}}\otimes\tilde\tau_{F'i}\oplus 0_{R'})+ \beta  (\ket{\psi_{j}}\bra{\psi_{j}}\otimes\tilde\tau_{F'j}\oplus 0_{R'})
\end{align*}
for all $i,j$ and $\alpha,\beta>0$ such that $\alpha+\beta=1$. We next show that $\tilde\tau_{F'i}=\tilde\tau_{F'j}$ for all $i\neq j$. To this aim, let us define $\rho_{ij}=(a_{i} \ket{\psi_{i}}+a_{j} \ket{\psi_{i}})(a_{i} \bra{\psi_{i}}+a_{j} \bra{\psi_{i}})$, with $a_{i},a_{j}\in\mathbb{C}$ such that $|a_{i}|^{2}+|a_{j}|^{2}=1$. Since $\rho_{ij}$ is pure, in view of the first condition in \eqref{eq:compsubpure}, it holds
\begin{align*}
	\Pi_{L'F'}(\rho_{ij}\otimes\tau_{F''}\oplus 0_{R''})\Pi_{L'F'} = \tilde{\rho}'_{ij}\simeq(\rho_{ij}\otimes\tilde\tau_{F'ij}\oplus 0_{R'}) ,
\end{align*}
so that 
\begin{align}\label{eq:ab}
	\Pi_{L'F'}\Pi_{k}(\rho_{ij}\otimes\tau_{F''}\oplus 0_{R''})\Pi_{k}\Pi_{L'F'} \simeq |a_{k}|^{2} (\rho_{k}\otimes\tilde\tau_{F'ij}\oplus 0_{R'}),
\end{align}
where $\Pi_{k}\simeq(\ket{\psi_{k}}\bra{\psi_{k}}\otimes I_{F''}\oplus0_{R''})$, $k=i,j$.
On the other hand, in view of the definition of $\rho_{ij}$, it also holds
\begin{align}\label{eq:ab2}
	\Pi_{L'F'}\Pi_{k}(\rho_{ij}\otimes\tau_{F''}\oplus 0_{R''})\Pi_{k}\Pi_{L'F'} \simeq  |a_{k}|^{2}(\rho_{k}\otimes\tilde\tau_{F'k}\oplus 0_{R'}),\quad k=i,j.
\end{align}
Eventually, a comparison of Eqs. \eqref{eq:ab} and \eqref{eq:ab2} yields the desired conclusion, that is, 
\(
\tilde{\tau}_{F'ij}=\tilde\tau_{F'i}=\tilde\tau_{F'j},$ for all $i,j, \,i\neq j.\)
We can apply the same argument as before to show that if the second condition in \eqref{eq:compsubpure} holds true for all pure $\rho_{L}$, then the same condition holds true for every, pure or mixed, $\rho_{L}$ with $\tilde{\tau}_{F''}(\rho_{L})$ independent of $\rho_{L}$. Hence, we proved that $\tilde{\tau}_{F'}(\rho_{L}),\, \tilde{\tau}_{F''}(\rho_{L})$ in Eq. \eqref{eq:compsubpure} are actually independent of $\rho_L$, or, equivalently, that the two initializations are compatible.
\qed\end{proof}

\subsection{Proof of Theorem 1}
\label{app:main}

\begin{proof}
We first address necessity, i.e., ``$\exists$ $\Phi_P$ that tolerates $\cN$ $\implies$ the initializations $\rho_{P'}$ and $\rho_{P'',\lambda}$ are compatible for all $\lambda\in[0,1]$''. To this aim, we first recall that if there exists  $\Phi_P$ correcting $\cN$ then $\Phi_P$ also corrects any convex combination of the noise and the identity operator, namely $\cM_{\lambda}=\lambda\cN+(1-\lambda)\cI$, $\lambda\in[0,1]$ (Lemma \ref{lem:convex}). Let $\{\ket{\psi_{i}}\}_{i}$ be any orthonormal basis in $\cH_{L}$ and $\rho_{P',i}\simeq \ket{\psi_{i}}\bra{\psi_{i}}\otimes\tau_{F'}\oplus 0_{R'}$. We will prove that for all $\lambda\in[0,1]$,
\begin{align}\label{eq:trij}
	\tr(\rho_{P',i}\cM_{\lambda}(\rho_{P',j}))=0,\quad \forall\,i,j,\ i\neq j.
\end{align}
Notice that if \eqref{eq:trij} holds true, then
\begin{align*}
	& \rho_{P',i}\cM_{\lambda}(\rho_{P',j}) \rho_{P',i}=0,\quad \forall\,i,j,\ i\neq j,
\end{align*}
which in turn implies
\begin{align*}
	\Pi_{L'F'}\cM_{\lambda}(\rho_{P',j})\Pi_{L'F'} = \rho_{P',j},\quad \forall\, j,
\end{align*}
where $\Pi_{L'F'}\simeq I_{L'}\otimes  I_{F'}\oplus 0_{R'}$ is the orthogonal projection onto $\cH_{L'}\otimes \cH_{F'}$.

In order to prove Eq. \eqref{eq:trij} pick any $\lambda\in[0,1]$ and assume, by contradiction, that there exists a pair $(i, j)$, $i\neq j$, such that 
$	\tr(\rho_{P',i}\cM_{\lambda}(\rho_{P',j}))\neq 0.$ The latter condition implies that $$\cM_{\lambda}(\rho_{P',j})\geq \rho_{P',j}^{(\gamma)} \simeq \ket{\psi_{i}}\bra{\psi_{i}}\otimes\gamma\oplus 0_{R'},$$ for suitable $\gamma\geq 0$ (positive semidefinite but not necessarily of unit trace). Notice also that $\rho_{P',i}\geq c\rho_{P',j}^{(\gamma)}$ for suitable $c\in\mathbb{R}$, $c>0$, since $\tau_{F'}$ is of full rank. Moreover, since $\Phi_P$ is a positive map, it preserves the partial ordering, i.e. $\Phi_P(X)\geq \Phi_P(Y)$ if $X\geq Y\geq 0$. Therefore, we have
\begin{align*}
	\tr(\Phi_P(\rho_{P',i})\Phi_P(\cM_{\lambda}(\rho_{P',j}))) & \geq \tr(\tilde{\Phi}_P(\rho_{P',i})\tilde{\Phi}_P(\cM_{\lambda}(\rho_{P',j})))\\
							          & \geq  c\,\tr(\tilde{\Phi}_P(\rho_{P',j}^{(\gamma)})^{2})>0,
\end{align*}
where $\tilde{\Phi}_P(X)=\Phi_P(\Pi_i X \Pi_i)$, $\Pi_i\simeq \ket{\psi_{i}}\bra{\psi_{i}}\otimes  I_{F'} \oplus 0_{R'}$. Hence, we proved that, for $i\neq j$,
\begin{align}\label{eq:cEneq}
	\tr(\Phi_P(\rho_{P',i})\Phi_P(\cM_{\lambda}(\rho_{P',j})))\neq 0.
\end{align} 
Now, since $\Phi_P$ tolerates $\cN$ and, therefore, $\cM_{\lambda}$ for $\lambda\in[0,1]$ (Lemma \ref{lem:convex}), we have $\Phi_P(\cM_{\lambda}(\rho_{P',j}))=\Phi_P(\rho_{P',j})$. In view of the latter fact and of Eq. \eqref{eq:cEneq}, it follows that
\( \tr(\Phi_P(\rho_{P',i})\Phi_P(\rho_{P',j}))\neq 0,\)  \(i\neq j, \)
which, in turn, implies that $\Phi_P$ is not orthogonality-preserving. This contradicts the fact that $\Phi_P$ must be an isometry \cite[Prop. 2.1]{busch}. The same argument outlined above holds if we consider $\rho_{P'',\lambda,i}\simeq \ket{\psi_{i}}\bra{\psi_{i}}\otimes\tau_{F'',\lambda}\oplus 0_{R''}$, and replace Eq. \eqref{eq:trij} with
\begin{align}
\label{eq:trijnew}
	\tr(\rho_{P'',\lambda,i}\cR(\rho_{P'',\lambda,j}))=0,\quad \forall\,i,j,\ i\neq j,
\end{align}
where $\cR$ is an isometric map from $\cH_{P}$ to $\cH_{P}$, mapping $\rho_{P'',\lambda}$ to $\rho_{P'}$.
Consequently, we have that both Eqs. \eqref{eq:trij} and \eqref{eq:trijnew} hold. These condition in turn imply that, for every basis $\{\ket{\psi_{i}}\}_{i}$,
\begin{align*}
	 \Pi_{L'F'}\rho_{P'',\lambda,j}\Pi_{L'F'} &\simeq \ket{\psi_{j}}\bra{\psi_{j}}\otimes\tilde{\tau}_{F'j}\oplus 0_{R'},\quad \forall\, j,\\
	 \Pi_{L''F''}\rho_{P',j}\Pi_{L''F''} &\simeq \ket{\psi_{j}}\bra{\psi_{j}}\otimes\tilde{\tau}_{F'',\lambda j}\oplus 0_{R''},\quad \forall\, j.
\end{align*}
In view of this and of Lemma \ref{lem:basis}, we conclude that the initializations $\rho_{P'}$ and $\rho_{P'',\lambda}$ in the subsystem decompositions  $\cH_{L'}\otimes \cH_{F'}\oplus \cH_{R'}$ and $\cH_{L''}\otimes \cH_{F'',\lambda}\oplus \cH_{R'',\lambda}$ are compatible for all $\lambda\in[0,1]$.

\medskip

We next address sufficiency, i.e., ``the initializations $\rho_{P'}$ and $\rho_{P'',\lambda}$ are compatible for all $\lambda\in[0,1]$ $\implies$ $\exists$ $\Phi_P$ that corrects $\cN$''. To this aim, pick $\hat\lambda\in(0,1)$ and consider the CPTP map
\begin{align}\label{eq:Elambda}
	\Phi_{\hat\lambda}\equiv \cI_{L''}\otimes \cF_{P'',\hat\lambda}\oplus \cI_{R'',\hat\lambda}
\end{align}
such that $\cF_{P'',\hat\lambda}(\sigma)=\tau_{F}$ for all $\sigma\in\cD(\cH_{F'',\hat{\lambda}})$. The previous map is an isometry which tolerates the map $\cM_{\hat{\lambda}}=\hat\lambda\cN+(1-\hat\lambda)\cI$, that is $\Phi_{\hat\lambda}(\rho_{P'',\hat\lambda})\simeq \rho_{L}\otimes \tau_{F}\oplus 0_{R}$. We want to show that $\Phi_{\hat{\lambda}}$ tolerates also the noise $\cN$, that is $\Phi_{\hat\lambda}(\rho_{P''}) \simeq \rho_{L}\otimes \tau_{F}\oplus 0_{R}$. We first notice that, since $\hat\lambda\in(0,1)$, we have \(\Pi_{L'F'}\leq \Pi_{L''F'',\hat{\lambda}},\) where $\Pi_{L''F'',\hat\lambda}$ denotes the orthogonal projection onto $\cH_{L''}\otimes \cH_{F'',\lambda}$. In view of this fact and since $\rho_{P'}$ and $\rho_{P'',\hat\lambda}$ are compatible by assumption, it holds
\begin{align*}
\Pi_{L''F'',\hat\lambda}\rho_{P'}\Pi_{L''F'',\hat\lambda}\simeq \rho_{L}\otimes \tilde{\tau}_{F'',\hat\lambda}\oplus 0_{R'',\hat\lambda},\quad \forall \rho_{L}\in\cD(\cH_{L}).
\end{align*}
Hence, we have
$\Phi_{\hat\lambda}(\rho_{P'})\simeq\rho_{L}\otimes \tau_{F}\oplus 0_{R},$ $\forall \rho_{L}\in\cD(\cH_{L}).$
Eventually, by linearity of $\Phi_{\hat\lambda}$ and in view of the definition of $\rho_{P'',\hat\lambda}$,
\begin{align*}
\Phi_{\hat\lambda}(\rho_{P''})&=\Phi_{\hat\lambda}\Big(\frac{1}{\hat\lambda}\rho_{P'',\hat\lambda}-\frac{1-\hat\lambda}{\hat\lambda}\rho_{P'}\Big)
	=\frac{1}{\hat\lambda}\Phi_{\hat\lambda}(\rho_{P'',\hat\lambda})-\frac{1-\hat\lambda}{\hat\lambda}\Phi_{\hat\lambda}(\rho_{P'})\\
	&\simeq \rho_{L}\otimes \tau_{F}\oplus 0_{R},
\end{align*}
for all $\rho_{L}\in\cD(\cH_{L})$, whereby the desired conclusion follows. \qed
\end{proof}

\end{appendix}


\begin{thebibliography}{99}

\bibitem{Nielsen-Chuang}
M.~A.~Nielsen and I.~L.~Chuang, \emph{Quantum Computation and Quantum Information}. Cambridge University Press, 2000.

\bibitem{QECBook} 
D.~A.~Lidar and T.~A.~Brun (Eds.),
\emph{Quantum Error Correction}. Cambridge University Press, Cambridge, 2013.

\bibitem{Blume-Kohout2010}
R.~Blume-Kohout, H.~Khoon Ng, D.~Poulin, and L.~Viola, ``Characterizing the structure of preserved information in quantum processes,'' {\em Phys. Rev. Lett.}, vol. 100:030501, 2008; ``Information-preserving structures: A general framework for quantum zero-error information,'' {\em Phys. Rev. A}, vol. 82:062306, 2010.

\bibitem{dfs} 
P. Zanardi and M. Rasetti, ``Noiseless quantum codes,'' {\em Phys. Rev. Lett.}, vol. 79: 3306, 1997; D. Lidar, I. Chuang, and K. Whaley, ``Decoherence free subspaces for quantum computation,'' {\em ibid.} vol. 81:2594, 1998.

\bibitem{Knill2000}
E.~Knill, R.~Laflamme, and L.~Viola,
\newblock ``Theory of quantum error correction for general noise,'' {\em Phys. Rev. Lett.}, vol. 84:2525, 2000.

\bibitem{Preskill}
J. Preskill, ``Sufficient condition on noise correlations for scalable quantum computing,'' {\em Quant. Inf. Comput.}, vol. 13:181, 2013.

\bibitem{Gottesman1998}
D.~Gottesman, ``Class of quantum error-correcting codes saturating the quantum Hamming bound,'' {\em Phys. Rev. A}, vol. 54:1862, 1996; ``Theory of fault-tolerant quantum computation,'' {\em ibid.}, vol. 57:127, 1998.

\bibitem{Knill1997}
E.~Knill and R.~Laflamme, 
\newblock ``Theory of quantum error-correcting codes,''
\newblock {\em Phys. Rev. A}, vol. 55:900, 1997.

\bibitem{Oqec} 
D. Kribs, R. Laflamme, and D. Poulin, ``Unified and generalized approach to quantum error correction,'' {\em Phys. Rev. Lett.}, vol. 94:180501, 2005; D. W. Kribs, R. Laflamme, D. Poulin, and M. Lesosky, ``Operator quantum error correction,'' {\em Quant. Inf. Comput.}, vol. 6:383, 2006.

\bibitem{Poulin2005}
D.~Poulin, ``Stabilizer formalism for operator quantum error correction,'' {\em Phys. Rev. Lett.}, vol. 95:230504, 2005.

\bibitem{Aliferis} 
P. Aliferis and A. W. Cross, ``Subsystem fault tolerance with the Bacon-Shor code,'' {\em Phys. Rev. Lett.}, vol. 98:220502, 2007.

\bibitem{Hideo2013}
G. Sarma and H. Mabuchi, ``Gauge subsystems, separability and robustness in autonomous quantum memories,'' {\em New J. Phys.}, vol. 15:035014, 2013.

\bibitem{Cleve}
R. Cleve and D. Gottesman, ``Efficient computations of encodings for quantum error correction,'' {\em Phys. Rev. A}, vol. 56:76, 1997.

\bibitem{Klappenecker2009}
P.~Sarvepalli and A.~Klappenecker,
\newblock ``Encoding subsystem codes,''
\newblock {\em Int. J. Adv. Security}, vol. 2:142, 2009.

\bibitem{altafini-introduction}
C.~Altafini and F.~Ticozzi, 
\newblock ``Modeling and control of quantum systems: An introduction,''
\newblock {\em IEEE Trans. Autom. Control}, vol. 57:1898, 2012.

\bibitem{kraus} 
B.~Kraus, S.~Diehl, A.~Micheli, A.~Kantian, H.~P.~B\"uchler, and P.~Zoller, ``Preparation of entangled states by quantum {M}arkov processes,'' \emph{Phys. Rev. A}, vol. 78:042307, 2008.

\bibitem{Ticozzi2012}
F.~Ticozzi and L.~Viola, ``Stabilizing entangled states with quasi-local quantum dynamical semigroups,'' {\em Phil. Trans. R. Soc. A}, vol. 370:5259, 2012; ``Steady-state entanglement by engineered quasi-local {M}arkovian dissipation,'' \emph{Quantum Inf. Comput.}, vol. 14:0265, 2014.

\bibitem{BTV2012}
G.~Baggio, F.~Ticozzi, and L.~Viola, 
\newblock ``Quantum state preparation by controlled dissipation in finite time: From classical to quantum controllers,'' 
\newblock {\em 2012 IEEE 51st IEEE Conf. Decision and Control (CDC)}, pp. 1072--1077, 2012.

\bibitem{cooling}
F.~Ticozzi and L.~Viola, ``Quantum resources for purification and cooling: Fundamental limits and opportunities,'' {\em Sci. Rep.}, vol. 4:5192, 2014.

\bibitem{opensys}
S. Lloyd and L. Viola, ``Engineering quantum dynamics,'' {\em Phys. Rev. A}, vol. 65:010101, 2001; J. T. Barreiro, M. M\"{u}ller, P. Schindler, D. Nigg, T. Monz, M. Chwalla, M. Hennrich, C. F. Roos, P. Zoller, and R. Blatt, ``An open-system quantum simulator with trapped ions,'' {\em Nature}, vol. 470:486, 2011.

\bibitem{clerk}
A. Metelmann and A. A. Clerk, ``Quantum-limited amplification via reservoir engineering,'' {\em Phys. Rev. Lett.}, vol. 112:133904, 2014.

\bibitem{Pastawski2011}
F.~Pastawski, L.~Clemente, and J.~I.~Cirac,
\newblock ``Quantum memories based on engineered dissipation,''
\newblock {\em Phys. Rev. A}, vol. 83:012304, 2011.

\bibitem{Verstraete-DQC}
F.~Verstraete, M.~M. Wolf, and J.~I. Cirac, ``Quantum computation and quantum-state engineering driven by dissipation,'' {\em Nat. Phys.}, vol. 5:633, 2009.

\bibitem{Reiter}
F.~Reiter, A.~S.~S{\o}rensen, P.~Zoller, and C.~A.~Muschik, ``Dissipative quantum error correction and application to quantum sensing with trapped ions,'' {\em Nat. Commun.}, vol. 8:1822, 2017.

\bibitem{Home}
B. de Neeve, Th. L. Nguyen, T. Behrle, and J. Home, ``Error correction of a logical grid state qubit by dissipative pumping,'' arXiv:2010.09681, 2020.

\bibitem{Sarlette}
V. Martin and A. Sarlette, ``Scaling up reservoir engineering for error-correcting codes,'' arXiv:2010.02850, 2020.

\bibitem{Dengis2014}
J.~Dengis, R.~K\"{o}nig, and F.~Pastawski, 
\newblock ``An optimal dissipative encoder for the toric code,'' 
\newblock {\em New J. Phys.}, vol. 16:013023, 2014.

\bibitem{Kitaev2003}
A.~Kitaev, 
\newblock ``Fault-tolerant quantum computation by anyons,''
\newblock {\em Ann. Phys.}, vol. 303:2, 2003.

\bibitem{johnson-FTS}
P.~D. Johnson, F.~Ticozzi, and L.~Viola, ``Exact quasi-local stabilization of entangled states in finite time,'' {\em Phys. Rev. A}, vol. 96:012308, 2017.

\bibitem{encoding-CDC} F. Ticozzi, G. Baggio and L. Viola, 
\newblock ``Quantum information encoding from stabilizing dynamics,'' 
\newblock {\em 2019 IEEE 58th Conf. Decision and Control (CDC),} Nice, France, pp. 413-418, 2019.

\bibitem{ShorCode}
P. W. Shor, ``Scheme for reducing decoherence in quantum computer memory,'' {\em Phys. Rev. A}, vol. 52:R2493, 1995.

\bibitem{SteaneCode} 
A. M. Steane, ``Error correcting codes in quantum theory,'' {\em Phys. Rev. Lett.}, vol. 77:793, 1996; ``Multiple-particle interference and quantum error correction,'' {Proc. Roy. Soc. Lond. A}, vol. 452:2551, 1996.

\bibitem{Viola2001}
L.~Viola, E.~Knill, and R.~Laflamme, 
\newblock ``Constructing qubits in physical systems,''
\newblock {\em J. Phys. A: Math. Theor.}, vol. 34:7067, 2001.

\bibitem{Knill2006}
E.~Knill,
``Protected realizations of quantum information,'' {\em Phys. Rev. A}, vol. 74:042301, 2006.

\bibitem{Ticozzi2010}
F.~Ticozzi and L.~Viola, 
\newblock ``Quantum information encoding, protection, and correction from trace-norm isometries,''
\newblock {\em Phys. Rev. A}, vol. 81:032313, 2010.

\bibitem{PerfectCode}
R. Laflamme, C. Miquel, J. P. Paz, and W. H. Zurek, ``Perfect quantum error correcting code,'' {\em Phys. Rev. Lett.}, vol. 77:198, 1996; C. H. Bennett, D. P. DiVincenzo, J. A. Smolin, and W. K. Wootters, ``Mixed state entanglement and quantum error correction,'' {\em Phys. Rev. A}, vol. 54:3824, 1996.

\bibitem{Breuckmann-thesis}
N. P. Breuckmann, {\em Quantum Subsystem Codes, Their Theory and Use}, Bachelor Thesis, 
Aachen University, 2011.

\bibitem{Bacon}
D. Bacon, S. T. Flammia, A. W. Harrow, and J. Shi, ``Sparse quantum codes from quantum circuits,''
{\em IEEE Trans. Inf. Theory}, vol 63:2464, 2017. 

\bibitem{AlickiLendi} 
R. Alicki and K. Lendi, {\em Quantum Dynamical Semigroups and Applications}. 
Springer-Verlag, Berlin, 2007. 

\bibitem{Kosloff} 
R. Kosloff, 
``Quantum thermodynamics: A dynamical viewpoint,'' {\em Entropy}, vol. 15:2100, 2013.

\bibitem{ticozzi-QDS}
F.~Ticozzi and L. Viola,
\newblock ``Quantum {M}arkovian subsystems: Invariance, attractivity and control,''
\newblock {\em IEEE Trans. Autom. Control}, vol. 53:2048, 2008.

\bibitem{ticozzi-markovian}
\newblock F. Ticozzi and L. Viola,
\newblock ``Analysis and synthesis of attractive quantum Markovian dynamics,''
\newblock{\em Automatica}, vol. 45:2002, 2009.

\bibitem{ticozzi-discretefeedback}
S.~Bolognani and F.~Ticozzi, 
\newblock ``Engineering stable discrete-time quantum dynamics via a canonical {QR} decomposition,'' 
\newblock {\em IEEE Trans. Autom. Control}, vol. 55:2721, 2010.

\bibitem{ticozzi-alternating}
F.~Ticozzi, L.~Zuccato, P.~D.~Johnson, and L.~Viola,
\newblock ``Alternating projections and discrete-time stabilization of quantum states,''
\newblock {\em IEEE Trans. Autom. Control}, vol. 63:819, 2018.

\bibitem{Dennis2002}
E.~Dennis, A.~Kitaev, A.~Landahl, and J.~Preskill,
\newblock ``Topological quantum memory,''
\newblock {\em J. Math. Phys.}, vol. 43:4452, 2002.

\bibitem{Pastawski2014}
R. K\"onig and F. Pastawski, ``Generating topological order: No speedup by dissipation,'' {\em Phys. Rev. B}, vol. 90:045101, 2014.

\bibitem{busch} P.~Busch, ``Stochastic isometries in quantum mechanics,'' {\em Math. Phys., Anal. Geom.}, vol. 2:1, 1999. 

\bibitem{Lucia}
A. Lucia, T. S. Cubitt, S. Michalakis, and D. P\'{e}rez-Garc\'{i}a, 
``Rapid mixing and stability of quantum dissipative systems,'' {\em Phys. Rev. A}, vol. 91:040302(R), 2015.

%

\bibitem{Liang}
C. Shen, K. Noh, V. V. Albert, S. Krastanov, M. H. Devoret, R. J. Schoelkopf, S. M. Girvin, and L. Jiang, ``Quantum channel construction with circuit quantum electrodynamics,'' {\em Phys. Rev. B}, vol. 95:134501, 2017.
\end{thebibliography}
\end{document}